\documentclass[9pt,twocolumn,twoside]{pnas-new}

\templatetype{pnasresearcharticle} 

\usepackage{amsthm}
\usepackage{bbm}
\usepackage[capitalize,nameinlink,noabbrev]{cleveref}
\usepackage{color}

\newtheorem{Theorem}{Theorem}

\newtheorem{Corollary}[Theorem]{Corollary}

\newcommand{\phylogeny}{{\cal T}}

\newcommand{\order}[1]{{\cal O}\hspace{-0.2em}\left( #1 \right)}

\newcommand{\bbR}{{\mathbb R}}
\newcommand{\bbN}{{\mathbb N}}
\newcommand{\bbC}{{\mathbb C}}

\newcommand{\bbP}{{\mathbb P}}

\newcommand{\bbE}{{\mathbb E}}

\newcommand{\Var}{\textnormal{Var\hspace{0.5mm}}}

\newcommand{\X}{\textbf{X}}

\newcommand{\Y}{\textbf{Y}}

\newcommand{\E}{{\mathbf{ E}}}

\newcommand{\diag}{\textnormal{diag}}

\newcommand{\op}{\textnormal{op}}
\newcommand{\as}{\textnormal{a.s.}}

\newcommand{\GD}[1]{#1}

\newcommand{\y}{\mathbf{y}}
\newcommand{\QQ}{\mathbf{Q}}
\newcommand{\LLambda}{\boldsymbol{\Lambda}}
\newcommand{\M}{\mathbf{M}}
\newcommand{\JJ}{\mathbf{J}}
\newcommand{\II}{\mathbf{I}}
\newcommand{\AAA}{\mathbf{A}}

\newcommand{\Zero}{\boldsymbol{0}}
\newcommand{\ttheta}{\boldsymbol{\theta}}
\newcommand{\mom}{\boldsymbol{\xi}}
\newcommand{\ppi}{\boldsymbol{\pi}}
\renewcommand{\P}{\mathbf{P}}
\newcommand{\p}{\mathbf{p}}
\newcommand{\q}{\mathbf{q}}
\newcommand{\B}{\mathbf{B}}
\newcommand{\Z}{\mathbf{E}}

\newcommand{\TT}{\mathbf{T}}
\newcommand{\RM}{\mathbf{R}}
\newcommand{\UU}{\mathbf{U}}

\newcommand{\RR}{\mathbb{R}}

\newcommand{\one}{\mathbf{1}}

 \newtheorem{@definition}{\bf Definition}
 \newenvironment{definition}{\begin{@definition}\rm}{\end{@definition}}

 \newtheorem{@remark}{\bf Remark}
 \newenvironment{remark}{\begin{@remark}\rm}{\end{@remark}}

 \newtheorem{@example}{\bf Example}

 \newtheorem{theorem}{\bf Theorem}
 \newtheorem{lemma}{\bf Lemma}

\allowdisplaybreaks[1]

\graphicspath{{figures/}}

\title{On the surprising effectiveness of a simple matrix exponential derivative approximation, with application to global SARS-CoV-2}

\author[a]{Gustavo Didier}
\author[a]{Nathan E.~Glatt-Holtz} 
\author[b,1]{Andrew J.~Holbrook}
\author[b]{Andrew F.~Magee}
\author[b,c,d]{Marc A.~Suchard}

\affil[a]{Department of Mathematics, Tulane University}
\affil[b]{Department of Biostatistics, University of California, Los Angeles}
 \affil[c]{Department of Biomathematics, University of California, Los Angeles}
 \affil[d]{Department of Human Genetics, University of California, Los Angeles}

\leadauthor{Didier} 

\significancestatement{Recent work uses a simplistic approximation to
  the matrix exponential derivative to apply gold-standard models from
  evolutionary biology to a collection of challenging data
  analyses. Whereas one may expect the naive approach to break down
  with increasing model dimensionality, empirical results show no such
  failure.  Here, we (1) develop rigorous error bounds that
  improve---in a certain sense---as model dimension grows and (2)
  demonstrate the scalability of the naive approach to a
  higher-dimensional analysis of the global spread of the virus
  responsible for the COVID-19 pandemic.}

 \authorcontributions{AJH and MAS designed research. GD, NEGH and AJH developed mathematical results. AJH and AFM performed simulations and data analysis.  GD, NEGH and AJH wrote the paper.}
\correspondingauthor{\textsuperscript{1}To whom correspondence should be addressed. E-mail: aholbroo@g.ucla.edu}


\begin{abstract}
  The continuous-time Markov chain (CTMC) is the mathematical
  workhorse of evolutionary biology.  Learning CTMC model parameters
  using modern, gradient-based methods requires the derivative of the
  matrix exponential evaluated at the CTMC's infinitesimal generator
  (rate) matrix.  Motivated by the derivative's extreme computational
  complexity as a function of state space cardinality, recent work
  demonstrates the surprising effectiveness of a naive, first-order
  approximation for a host of problems in computational biology.  In
  response to this empirical success, we obtain rigorous deterministic and
  probabilistic bounds for the error accrued by the naive
  approximation and establish a ``blessing of dimensionality'' result
  that is universal for a large class of rate matrices with random
  entries.  Finally, we apply the first-order approximation within
  surrogate-trajectory Hamiltonian Monte Carlo for the analysis of the
  early spread of SARS-CoV-2 across 44 geographic regions that
  comprise a state space of unprecedented dimensionality for
  unstructured (flexible) CTMC models within evolutionary biology.
\end{abstract}

\begin{document}

\maketitle
\thispagestyle{firststyle}
\ifthenelse{\boolean{shortarticle}}{\ifthenelse{\boolean{singlecolumn}}{\abscontentformatted}{\abscontent}}{}


\dropcap{P}hylogeographic methods
\cite{lemey2009bayesian, lemey2014unifying, holbrook2021massive,
  holbrook2022viral} model large-scale viral transmission between
human populations as a function of the shared evolutionary history of
the viral population of interest.  Data take the form of dates,
locations and genome sequences associated to individual viral
samples. Spatiotemporal structure interfaces with network structure
given by the phylogeny, or family tree, describing the viruses'
collective history beginning with the most recent common ancestor.
While one cannot directly observe this history, one may statistically
reconstruct the phylogenetic tree by positing that changes in the
viral genome happen randomly at regular intervals, thereby capturing
the intuition that viral samples with more differences between their
(aligned) sequences should find themselves further apart on the family
tree.

The continuous-time Markov chain (CTMC) \cite{norris_1997}
represents the gold-standard mathematical model for such evolution of
characters (e.g., nucleotides) within a fixed span of evolutionary
time. A CTMC defined over a discrete, $d$-element state space consists
of a row vector $\ppi_0$ whose individual components describes the
probability of inhabiting each of the possible states at time $t=0$,
as well as a $d\times d$ infinitesimal generator (or rate) matrix
$\QQ$ with non-negative off-diagonal elements $q_{ij}$, $i\neq j$, and
non-positive diagonal elements $q_{ii}=-\sum_j q_{ij}$. For any lag
$t\geq 0$, the matrix exponential
\cite{moler1978nineteen,moler2003nineteen} provides the Markov chain's
transition probability matrix
\begin{align}
    \P_t := e^{t \QQ} := \sum_{n=0}^\infty \frac{t^n\QQ^n}{n!}  \, ,
\end{align}
which has elements $[\P_t]_{ij}$ that dictate the probability of the
process jumping from state $i$ to state $j$ after time $t$. It is
straightforward to verify that $\P_t$ is a valid transition matrix,
having probability vectors for rows: if $\one$ and $\Zero$ are the
column vectors of ones and zeros, respectively, then $\QQ\one=\Zero$
and, therefore, $\P_t\one=\one$. The law of total probability then
provides the marginal probability of the process at any time
$t \geq 0$ as $\ppi_t=\ppi_0 e^{t\QQ}$.

Whether frequentist \cite{felsenstein1981evolutionary} or Bayesian
\cite{sinsheimer1996bayesian, yang1997bayesian, mau1999bayesian,
  suchard2001bayesian}, likelihood-based approaches to phylogenetic
reconstruction allow phylogenetic tree branch lengths to parameterize
time lags within the CTMC framework. We present the exact statement of
the phylogenetic CTMC paradigm below (see \cref{sec:sars:cov:app}).
Here, we note that the historical importance of tree-reconstruction
from aligned sequences leads to an early emphasis on the
sparse specification of $\QQ$ based on biologically motivated
assumptions
\cite{jukes1969evolution,kimura1980simple,hasegawa1985dating}.
Classical Markov chain Monte Carlo (MCMC) procedures
\cite{metropolis1953equation,hastings1970monte} work well for such
low-dimensional models.  But the phylogenetic CTMC framework has
applications beyond simple nucleotide substitution models. Within,
e.g., Bayesian phylogeography, the work in \cite{lemey2009bayesian}
provides a phylogenetic CTMC model for the spread of avian influenza
across $d=20$ global geographic locations but, for computational
reasons, favors a low-dimensional $\order{d}$ parameterization of
$\QQ$. Similarly, \cite{lemey2014unifying} model the spread of
influenza A H1N1 and H3N2 between as many as $d=26$ geographic regions
but---again for computational reasons---fit the model with
approximation techniques that provide no inferential guarantees.

Recently, \cite{magee2023random} demonstrate the feasibility of
approximate gradient-based methods for both maximum \emph{a
  posteriori} and full Bayesian inference of flexible and
fully-parameterized rate models and apply these methods to a
gold-standard $\order{d^2}$ mixed-effect CTMC model for the spread of
A H3N2 influenza between $d=14$ geographic locations. The usual
CTMC log-likelihood gradient calculations feature the matrix
exponential derivative (see \eqref{eq:sen:eq}, \eqref{eq:var:const:2} below)
\begin{align}
\label{eq:grad:def}
  &\nabla_{\JJ} e^{t \QQ}
    := \lim_{\epsilon \to 0}
    \frac{e^{t (\QQ+ \epsilon \JJ)} - e^{t\QQ} }{\epsilon} \\
  &=  e^{t \QQ} \sum_{n = 0}^\infty \frac{t^{n+1}}{(n+1)!}  
    \left( \sum_{\ell = 0}^n
    (-1)^\ell \binom{n}{\ell} \QQ^\ell \JJ \QQ^{n -\ell}\right)
    \label{eq:grad:series}
\end{align}
computed in the direction $\JJ$ of each of the $d^2$ natural basis
elements $\JJ_{ij}$ spanning the space of real-valued, $d\times d$
matrices ${\mathcal M}(d)={\mathcal M}(d,\bbR)$, thereby requiring at least $\order{d^5}$
floating point operations \cite{najfeld1995derivatives}.

Within the phylogenetic CTMC models of Section \ref{sec:sars:cov:app}, log-likelihood derivative computations that require $\nabla_{\JJ} e^{t \QQ}$ balloon to $\order{KNd^5}$, for $N$ the number of biological specimens observed and $K$ the number of parameters parameterizing $\QQ$.
To address this overwhelming computational cost, \cite{magee2023random} leverage
the simplistic approximation obtained by setting $n=0$ within \eqref{eq:grad:series}:
\begin{align}\label{eq:firstOrder}
	\widetilde{\nabla}_{\JJ} e^{t\QQ} := t e^{t\QQ}\JJ   \, .
\end{align}
\cite{magee2023random} show that this approximation helps reduce total cost to $\order{Kd^2 + Nd^3}$ and use this speedup within
surrogate-trajectory Hamiltonian Monte Carlo (see \cref{sec:Sur:trad:HMC}) 
to obtain a
34-fold improvement in effective sample size per second (ESS/s) over
random-walk MCMC within their 14 region phylogeographic example.
When trying to explain the remarkable empirical performance of the naive
approximation, the authors derive an error upper
bound (for an arbitrary matrix norm)
\begin{align}\label{eq:initial:bnd}
  \lVert \widetilde{\nabla}_{\JJ}e^{t \QQ}  - \nabla_{\JJ} e^{t \QQ} \rVert
  \leq \frac{ \| \JJ \| \lVert\QQ \rVert }{2} (e^{2t} -2t-1) 
\end{align}
that fails to leverage the specific forms of $\QQ$ and $\JJ$.
Notably, this bound explodes as either $t$ or $\lVert \QQ \rVert$
diverges to $\infty$.  Of course, the latter quantity would be
expected to grow large with dimension $d$ without more careful
structural assumptions, e.g., that $\QQ$ is a rate matrix belonging to
the class
\begin{align}\label{eq:rate:mat}
  \mathcal{R}(d) :=
  \Big\{ \QQ \in {\mathcal M}(d) |
  \QQ_{ij} \geq 0 \text{ for } j \not= i ,\, \sum_{j =1}^d \QQ_{ij} = 0 \Big\} \, .
\end{align}

In the following, we use the finer structural properties of $\QQ$ to
obtain more precise bounds on
\begin{align}\label{eq:true:diff}
  \Z(t) := \nabla_\JJ e^{t \QQ} - t e^{t \QQ} \JJ.
\end{align}
In \cref{thm:ass:error}, we provide an affine (in $t$) correction to the 
approximation \eqref{eq:firstOrder} that yields an exponentially tight 
$t \to \infty$ asymptotic for the error \eqref{eq:true:diff}. Then, in 
\cref{t:lambda_n-1(Q)->-infty_non-symm_a.s.}, we establish
precise probabilistic bounds in the high-dimensional $d \to \infty$
limit for a large class of randomly drawn rate matrices
$\QQ \in \mathcal{R}$.  Here we show, for any $\QQ \in \mathcal{R}$
whose off-diagonal elements are determined by independently and identically distributed (iid) draws from a
positive, sub-exponential distribution $F$, that all of the non-zero
singular values grow as $\sigma_j(\QQ) \sim d$ along with
asymptotically valid almost sure bounds for these rates as
$d \to \infty$. 

In regards to this second result,
\cref{t:lambda_n-1(Q)->-infty_non-symm_a.s.}, note that random rate (or
`Laplacian') matrices have attracted a great deal of attention from the
probability research community, especially in regard to their
high-dimensional properties; see e.g., \cite{takahashi:1969,bai:1999,
  chafai:2010, bordenave:caputo:chafai:2012, chatterjee:hazra:2022,
  nakerst:denisov:haque:2023}. In particular, seminal papers such as
\cite{bryc:dembo:jiang:2006, ding:jiang:2010,
  bordenave:caputo:chafai:2014} establish broad characterizations of
bulk behavior for the Laplacian eigenspectrum.  One contribution of
this paper, of independent interest, is a short and self-contained
construction of useful bounds for the singular values of
$\QQ \in \mathcal{R}$.

In \cref{thm:bride:of:frankinstein} and
\cref{cor:rig:det:rand:wrapup}, we show how \cref{thm:ass:error} and
\cref{t:lambda_n-1(Q)->-infty_non-symm_a.s.} combine to provide a more
refined analysis of $\E$ for suitable randomly generated
$\QQ \in \mathcal{R}(d)$. In \cref{thm:bride:of:frankinstein}, we
establish that particular terms appearing in the bound
\eqref{eq:ass:error} in \cref{thm:ass:error} decay with a rate on the
order of $1/d$ in the operator norm topology for large $d$ for certain
classes of symmetric generators $\QQ$. This class includes the random
elements considered in \cref{t:lambda_n-1(Q)->-infty_non-symm_a.s.}.
Here, although \cref{cor:rig:det:rand:wrapup} applies only for
symmetric matrices composed of sub-exponential draws, we provide
strong supplemental numerical evidence that our bounds remain valid
well beyond this special symmetric, sub-exponential special case; see
\cref{rmk:better:approx} and \cref{fig:app}.

One notable practical implication of \cref{thm:bride:of:frankinstein},
\cref{cor:rig:det:rand:wrapup} and \cref{rmk:better:approx} is the
identification of a further correction to
\eqref{eq:firstOrder}. Crucially this correction has the same
$\order{d^3}$ computational cost as \eqref{eq:firstOrder} while
leading to an asymptotically temporally uniformly accurate
approximation of $\nabla_{\JJ} e^{t \QQ}$.  See
\cref{rmk:better:approx} and \eqref{eq:better:help:on:sale} below.

 \cref{sec:emp} contains simulation studies comparing: accuracy of matrix exponential derivative approximations for different distributional assumptions on the generator matrix; the posterior distributions obtained using surrogate-trajectory Hamiltonian Monte Carlo (\cref{sec:Sur:trad:HMC}) and traditional Hamiltonian Monte Carlo; and parameter identification under different priors on generator matrix elements.

In \cref{sec:sars:cov:app} we follow these theoretical and empirical investigations
with an application of the naive, first-order gradient approximation
\eqref{eq:firstOrder} to a challenge in phylogeography requiring
Bayesian inference of a rate-matrix of unprecedented dimensionality.
Namely, we apply the approximation to a gold-standard mixed-effects,
phylogenetic CTMC model that uses 1,897 parameters to describe the
spread of SARS-CoV-2 across a $d=44$ dimensional state space
consisting of different global geographic locations.  Such an
application complements the empirical studies of
\cite{magee2023random} in a manner that emphasizes the naive
approximation's potential for impact.

\section{Rigorous Results}
\label{sec:rig:results}

This section lays out our rigorous results \cref{thm:ass:error},
\cref{thm:bride:of:frankinstein},
\cref{t:lambda_n-1(Q)->-infty_non-symm_a.s.} and
\cref{cor:rig:det:rand:wrapup}, the proofs of which appear in the
supplemental material below.

In what follows we adopt the following notational conventions. For any
${\mathbf A} \in {\mathcal M}(d)$, we list the associated (not
necessarily distinct) eigenvalues of $\mathbf{A}$ in ascending order
according to their real part, namely,
\begin{align}\label{eq:EV:real:ass}
  \Re \lambda_1({\mathbf A}) \leq \cdots \leq \Re \lambda_d({\mathbf A}).
\end{align}
Similarly, the singular values of ${\mathbf A}$ are written in
ascending order
\begin{align}\label{eq:SV:ass}
  \sigma_1({\mathbf A}) \leq \cdots \leq \sigma_d({\mathbf A}).
\end{align}
We let ${\mathcal S}(d) = {\mathcal S}(d,\bbR)$ and ${\mathcal S}(d,\bbC)$ represent the spaces of $d \times d$ symmetric and Hermitian matrices,
respectively.

We make use of multiple matrix norms leading to
materially different bounds as $d \rightarrow \infty$ \cite{trefethen2022numerical}.  Take
\begin{align}\label{eq:frob:norm}
  \|\mathbf{A}\|_{F} := \sqrt{\sum_{i,j =1}^d \mathbf{A}_{ij}^2} =
  \sqrt{\sum_{j =1}^d \sigma_j(\mathbf{A})^2}
\end{align}
for the \emph{Frobenius norm} and
\begin{align}\label{eq:op:norm}
  \|{\mathbf A}\|_{\textnormal{op}} :=  \sqrt{\lambda_d({\mathbf
  A}^*{\mathbf A})} = \sqrt{\lambda_d({\mathbf A}{\mathbf A}^*)}
  = \sigma_d(\mathbf{A})
\end{align}
for the \emph{operator norm} of ${\mathbf A}$.  Finally note that when we
simply write $\| \cdot \|$ the statement then holds for any valid
matrix norm as in our formulation of \cref{thm:ass:error}.

\subsection*{Deterministic bounds on approximation
  error in time}

We begin by deriving a dynamical equation for the error $\Z(t)$, defined in
\eqref{eq:true:diff}. Recall that $\X(t) := e^{t \QQ}$ for any
$\QQ \in \mathcal{M}(d)$ obeys the (matrix-valued) ordinary
differential equation
\begin{align}
   \frac{d\X}{dt} = \QQ \X, \quad \X(0) = \II.
\end{align}
Setting
$\Y^\epsilon = \epsilon^{-1} ( e^{t (\QQ+ \epsilon \JJ)}- e^{t \QQ})$
and taking a limit as $\epsilon \to 0$ we find that
$\Y = \nabla_{\JJ} e^{t \QQ}$ obeys
\begin{align}
  \frac{d \Y}{dt}  = \QQ \Y + \JJ \X = \QQ \Y + \JJ e^{t \QQ},
  \quad \Y(0) = \Zero.
	\label{eq:sen:eq}
\end{align}
Thus, variation of constants yields that, for any $t \geq 0$,
\begin{align}
	\nabla_{\JJ} e^{t \QQ} &= e^{t \QQ} \int_0^t e^{-s\QQ} \JJ e^{s\QQ} ds \label{eq:var:const:1}\\
	&=  e^{t \QQ} \sum_{k,m = 0}^\infty \int_0^t \frac{(-s)^{k} \QQ^k \JJ s^{m}\QQ^m}{k!m!}ds
		\notag\\
	&=  e^{t \QQ} \sum_{n = 0}^\infty \frac{t^{n+1}}{(n+1)!}  
	\left( \sum_{\ell = 0}^n (-1)^\ell \binom{n}{\ell} \QQ^\ell \JJ \QQ^{n -\ell}\right).
	\label{eq:var:const:2}
\end{align}
Taking the first order ($n=0$) approximation in \eqref{eq:var:const:2} produces
\eqref{eq:firstOrder}.  Note that, from \eqref{eq:var:const:1}, this
approximation $\widetilde{\nabla}e^{t\QQ}\JJ$ is evidently exact in the
special case when $\JJ$ and $\QQ$ commute.

Next notice that, if we differentiate $\widetilde{\Y} := t e^{t \QQ} \JJ$ in $t$, we
find that $\widetilde{\Y}$ obeys
\begin{align}\label{eq:sen:eq:apx}
  \frac{d \widetilde{\Y}}{dt} = \QQ \widetilde{\Y} +  e^{t \QQ} \JJ,
  \quad \widetilde{\Y}(0) = \Zero.
\end{align}
Thus, taking the error $\Z(t)$ as in \eqref{eq:true:diff} and
combining \eqref{eq:sen:eq} with \eqref{eq:sen:eq:apx} yields
\begin{align}
\label{eq:Z:dym}
  \frac{d\Z}{dt} = \QQ \Z + \JJ e^{t \QQ} -  e^{t \QQ} \JJ,
  \quad \Z(0) = \Zero.
\end{align}
Hence, again integrating this expression, we find
\begin{align}
  \Z(t) &= e^{t \QQ} \int_0^t (e^{-s \QQ} \JJ e^{s \QQ} - \JJ) ds
 \label{eq:ass:procs:0}\\
	&=  e^{t \QQ} \sum_{n = 1}^\infty \frac{t^{n+1}}{(n+1)!}  
   \left( \sum_{\ell = 0}^n (-1)^\ell
   \binom{n}{\ell} \QQ^\ell \JJ \QQ^{n -\ell}\right)
   \label{eq:ass:procs}
\end{align}
as could also be directly deduced from \eqref{eq:var:const:1},
\eqref{eq:var:const:2}.

Given a rate matrix $\QQ$ in $\mathcal{R}(d)$, recall that
$\Re \lambda_{d-1}(\QQ) \leq \lambda_{d}(\QQ)= 0$ by
the Gershgorin circle theorem \cite[Theorem
6.1.1]{horn:johnson:2012}.  Imposing a further
non-degeneracy assumption (e.g., that $\Re \lambda_{d-1}(\QQ)< 0$) we
therefore have an exponential decay in $\QQ e^{t\QQ}$.  This starting
point suggests that, under fairly general conditions, we may decompose
\eqref{eq:ass:procs:0} into a component where $\QQ e^{t\QQ}$ induces
an exponential decay in time and a complementary component taking the
form of a time-affine correction term.

These observations lead to the following theorem, the proof of which appears
in the supplement as \cref{sec:Thm:det:main}.
\begin{theorem}\label{thm:ass:error}
  Suppose that $\QQ, \JJ \in \mathcal{M}(d)$ for some $d \geq 1$.  We
  assume that we can find an element $\QQ^+ \in \mathcal{M}(d)$ such that
  $\QQ$ is a generalized inverse of $\QQ^+$, namely,
  \begin{align}\label{eq:GI:prop}
    \QQ^+ \QQ \QQ^+ = \QQ^+
  \end{align}
  and such that
  \begin{align}\label{eq:MP:proj:prop}
    e^{\tau \QQ}(\II - \QQ^+ \QQ) \!=\!  \II - \QQ^+ \QQ, \;
    (\II - \QQ \QQ^+)e^{\tau \QQ} \!=\!  \II - \QQ \QQ^+,
  \end{align}
  for any $\tau \in \bbR$. Furthermore, we suppose that $\QQ^+$ and
  $\QQ$ commute
  \begin{align}\label{eq:GI:com}
    \QQ \QQ^+ = \QQ^+ \QQ.
  \end{align}
  Finally, taking $\|\cdot\|$ be any matrix norm, we assume that
  \begin{align}\label{eq:qetq:exp:decay}
    \| \QQ e^{\tau \QQ}\| \leq C_0 e^{-\kappa \tau},
    \quad
    \text{ for all } \tau \geq 0,
  \end{align}
  where the constants $C_0 > 0, \kappa > 0$
  are independent of $\tau$.  Then, under these circumstances, 
  \begin{align}\label{eq:ass:error}
    \|\QQ^+ \JJ (\II - \QQ \QQ^{+})-& (\II -\QQ^+ \QQ) \JJ \QQ^{+}
        +t(\II -\QQ^+ \QQ) \JJ \QQ \QQ^{+}
                \notag\\
      &
        +  \nabla_\JJ e^{t \QQ} - t  e^{t\QQ} \JJ\| \leq C(1+ t)e^{-\kappa t}
\end{align}
for any $t \geq 0$.  Here, $C > 0$ is a $t$-independent
constant which is given explicitly as
\begin{align}
  C_0 \big(\| (\II -\QQ^+ \QQ) \JJ (\QQ^{+})^2\|
  \!&+\! \|(\QQ^+)^2\JJ (\II - \QQ \QQ^{+})\|
    \!+\! \|\QQ^+ \hspace{0.5mm}\JJ\|\big)
    \notag\\
    &+ C_0^2 \| \QQ^+  \JJ  \QQ^{+} \|.
      \label{eq:exp:error:ass:er}
  \end{align}
\end{theorem}

\begin{remark}\label{rmk:mat:class:dcy}
  To illuminate the scope of \cref{thm:ass:error}, we have the
  following three classes of matrices maintaining the conditions
  \eqref{eq:GI:prop}--\eqref{eq:qetq:exp:decay} as follows.
  \begin{itemize}
  \item[(i)] Suppose that $\QQ \in \mathcal{M}(d)$ is such that
    \begin{align}\label{eq:ev:d:cond}
      \Re \lambda_{d-1}(\QQ) <  \Re \lambda_d({\mathbf Q}) \leq 0,
       \lambda_d({\mathbf Q})  \text{ is simple }
    \end{align}
    and such that, if $\lambda_d({\mathbf Q})$ has an imaginary component,
    then its real part is strictly negative.  Under these
    circumstances, writing $\QQ$ in its Jordan canonical form yields,
    for some $m \geq 1$,
    \begin{equation}\label{e:Q=Jordan_decomp}
      \QQ = \M \diag (J_{1},\hdots,J_{m-1},\lambda_d(\QQ))\M^{-1}.
    \end{equation}
    Here, under \eqref{eq:ev:d:cond} each of
    these blocks $J_j$ must be invertible and so
    we may take
    \begin{align}
      \QQ^+ := \M \,
      \diag (J_{1}^{-1},\hdots,J_{m-1}^{-1},\lambda_d(\QQ)^{+})
      \, \M^{-1},
    \end{align}
    where 
    \begin{align}\label{eq:eig:gen:inv}
      \lambda_d(\QQ)^{+} =
      \begin{cases}
        0& \text{ if } \lambda_d(\QQ) = 0,\\
        \lambda_d(\QQ)^{-1}& \text{ otherwise.}
       \end{cases}
    \end{align}
  \item[(ii)] We next consider the case where $\QQ \in \mathcal{M}(d)$
    is diagonalizable and its spectrum lies strictly on the left half
    plane or at the origin.  This time we can write
    \begin{align}\label{eq:spec:decomp}
    \QQ = \M \LLambda \M^{-1} \text{ where }
      \LLambda = \diag (\lambda_1(\QQ), \ldots ,\lambda_d(\QQ))
    \end{align}
    and we set 
    \begin{align}\label{eq:gen:inv:diag}
     \! \!\! \!\QQ^+ \! \! \!= \!\M \LLambda^+\M^{-1} \text{ with }
      \LLambda^{+} \! \! \!= \! \diag (\lambda_1(\QQ)^+\!\! \! \!, \ldots
      ,\lambda_d(\QQ)^+).
      \end{align}
    The complex numbers $\lambda_j(\QQ)^+$,
    $j = 1, \ldots, d$ are defined as in \eqref{eq:eig:gen:inv}.
  \item[(iii)] Finally, we specialize to the case where
    $\QQ \in \mathcal{S}(d)$ is symmetric. In this case,
    $\QQ = \UU \LLambda \UU^*$, where $\UU$ is a unitary matrix and
    $\LLambda = \diag (\lambda_1(\QQ), \ldots ,\lambda_d(\QQ))$,
    $\lambda_j(\QQ)$ are its (real) eigenvalues.  We suppose that
    these eigenvalues are all non-positive,
    \begin{align}\label{eq:ev:d:cond:sym:case}
      \lambda_d({\mathbf Q}) \leq 0,
    \end{align}
    Here we take $\QQ^+$ as the Moore-Penrose inverse, namely,
    \begin{align}\label{eq:MP:inv:sym:case}
      \QQ^+ = \UU \LLambda^+ \UU^*,
   \end{align}
    where $\LLambda^+$ is as in \eqref{eq:gen:inv:diag}.   
  \end{itemize}
\end{remark}
In anticipation of \cref{t:lambda_n-1(Q)->-infty_non-symm_a.s.} below
and our desired application in \cref{sec:sars:cov:app} we are
preoccupied with the dimensional dependence of the constants in
\eqref{eq:qetq:exp:decay}, \eqref{eq:ass:error}
and \eqref{eq:exp:error:ass:er} in our formulation of
\cref{thm:ass:error}.  We next provide some such desirable bounds in
case (iii) of \cref{rmk:mat:class:dcy}.  Note that analogous results
for generators $\QQ$ in the classes (i) or (ii) would seemingly
require a delicate analysis of the associated eigenspaces (i.e., of the
structure of $\M$ in \eqref{e:Q=Jordan_decomp} or
\eqref{eq:spec:decomp} respectively).  However,
\cref{rmk:non:sym:num:ev} and \cref{fig:app} provide numerical evidence
of a broader scope for dimensionally improving approximations beyond
the symmetric case, at least for certain classes of randomly drawn
matrices.
\begin{theorem}\label{thm:bride:of:frankinstein}
  Let symmetric $\QQ \in \mathcal{S}(d)$ be non-positive, i.e.,  suppose that
  \eqref{eq:ev:d:cond:sym:case} holds.  Take $\QQ^+$ as in
  \eqref{eq:MP:inv:sym:case} and define
    \begin{align}\label{eq:first:neg}
      d_{-} = \max\{ 1 \leq j \leq d | \Re \lambda_{j}(\QQ) < 0\}.
    \end{align}
    Then, for any $\JJ \in \mathcal{M}(d)$,
  \begin{align}
      \| t(\II&  -\QQ^+  \QQ) \JJ \QQ \QQ^{+}
                +  \nabla_\JJ e^{t \QQ} - t  e^{t\QQ} \JJ\|_{F} 
                \label{eq:bride:of:frankinstein:1}\\
    \leq&  
    \Big( \! \sqrt{d_{-}}|\lambda_1(\QQ)| \! \cdot\!  \Big[ 2 \sqrt{d-d_{-}} \| \QQ^+\|_F^2 +  \| \QQ^+\|_F
    \notag \\ 
              &\qquad \qquad \qquad \;
    + \sqrt{d_{-}}|\lambda_1(\QQ)|\| \QQ^+\|_F^2 \Big]   (1+t)e^{- t|\lambda_{d-}(\QQ)|}
   \notag \\ 
&\quad+  2\sqrt{d - d_{-}} \| \QQ^+\|_F  \Big)\|\JJ\|_{F},  \notag
  \end{align}
  with $\| \QQ^+\|_F^2 := \sum^{d_-}_{k=1}\frac{1}{\lambda^2_{k}(\QQ)}$, whereas
  \begin{align}
    \| t(&\II -\QQ^+ \QQ) \JJ \QQ \QQ^{+}
         +  \nabla_\JJ e^{t \QQ} - t  e^{t\QQ} \JJ\|_{\op}
         \label{eq:bride:of:frankinstein}\\
    \leq &
           \Big(
           \frac{|\lambda_1(\QQ)|(2 + |\lambda_{d_-}(\QQ)|  + |\lambda_1(\QQ)|)}{\lambda^2_{d-}(\QQ)} 
           (1+t) e^{-t |\lambda_{d_-}(\QQ)|}
           \notag
    \\ 
              &\qquad+ \frac{2}{|\lambda_{d_{-}}(\QQ)|}\Big) \|\JJ\|_{\op}.
                \notag
  \end{align}
  Under the further assumption that $\lambda_{d}(\QQ) = 0$
  and
  \begin{align}
    \mu_1 d \leq |\lambda_{d-1}(\QQ)|
    \leq |\lambda_{1}(\QQ)| \leq \mu_2 d
  \end{align}
  for some $0< \mu_1\leq \mu_2$, we have
    \begin{align}\label{eq:bride:of:frankinstein:1:rm:app}
    \| t(&\II  -\QQ^+ \QQ) \JJ \QQ \QQ^{+}
         +  \nabla_\JJ e^{t \QQ} - t  e^{t\QQ} \JJ\|_{F}
         \\
     \leq &    \Big(\frac{\mu_2}{\mu_1^2}  \cdot \Big[ 2 \sqrt{d} + \mu_1 d  
            +  \mu_2 d^2\Big](1+t)^{-t \mu_1 d} +  \frac{2}{\mu_1 \sqrt{d}} \Big) \|\JJ\|_{F},
            \notag
  \end{align}
  and that
    \begin{align}\label{eq:bride:of:frankinstein:0:rm:app}
    \| t(\II & -\QQ^+ \QQ) \JJ \QQ \QQ^{+}
         +  \nabla_\JJ e^{t \QQ} - t  e^{t\QQ} \JJ\|_{\op}\\
      \leq &  \Big( \frac{\mu_2}{\mu^2_1}\cdot\Big[   \frac{2}{d} + \mu_1  +\mu_2\Big]
             (1+t)e^{-t \mu_1 d} + \frac{2}{\mu_1 d}\Big)
            \|\JJ\|_{\op}.\notag
  \end{align}
\end{theorem}

\subsection*{High-dimensional asymptotics
  via random matrix theory}

We turn to our probabilistic bounds on the singular values of randomly
generated rate matrices, \cref{t:lambda_n-1(Q)->-infty_non-symm_a.s.}.
Although interesting in its own right, this result leads to
consequences for the bounds in \eqref{eq:qetq:exp:decay} and
\eqref{eq:ass:error} when applied in \cref{thm:bride:of:frankinstein}.
Before proceeding, we briefly introduce further mathematical
preliminaries associated with the so-called sub-exponential random
variables.  To avoid confusion, note that the following definition uses the term in the same way as, e.g., \cite{vershynin:2018}, but that other definitions that mean quite the opposite (i.e., heavier than exponential tails) appear in the literature \cite{goldie1998subexponential}.
\begin{definition}
  \label{def:sub:exp:RV}
  A random variable $X$ is called \textit{sub-exponential} if there
  exists some constant $K > 0$ for which its tails satisfy
\begin{equation}\label{e:def_subexp_RV}
\bbP(|X|\geq t) \leq 2 e^{-t/K}, \quad \forall t \geq 0.
\end{equation}
In this case, the \textit{sub-exponential norm} of $X$ is defined by
\begin{equation}\label{e:sub-exp_norm}
\|X\|_{\psi_1} = \inf\big\{s > 0: \bbE e^{|X|/s} \leq 2\big\}.
\end{equation}
The \textit{class of sub-exponential distributions} is denoted by
\begin{align*}
L_{\psi_1} = \big\{F_X(dx): \|X\|_{\psi_1} < \infty \big\}.
\end{align*}
\end{definition}

\begin{remark}\label{rmk:exp:mom}
  In fact, by \cite[Proposition 2.7.1, p.\ 33]{vershynin:2018},
  condition \eqref{e:def_subexp_RV} is equivalent to the existence of
  some $s_0 > 0$ such that $\bbE e^{|X|/s_0} \leq 2$, namely,
  $\|X\|_{\psi_1} < \infty$ in \eqref{e:sub-exp_norm}.  As a notable
  example, if $X \sim \exp(\lambda)$, $\lambda > 0$, then it is easy
  to see that $X \in L_{\psi_1}$, indeed.
\end{remark}

We formulate our second major result as follows. 
\begin{theorem}\label{t:lambda_n-1(Q)->-infty_non-symm_a.s.}
  Let $F_X \in L_{\psi_1}$ be a distribution such that $X \geq 0$
  a.s., $\bbE X = \mu > 0$ and $\Var X = \sigma^2 > 0$. We consider a
  sequence of random matrices
  $\QQ \equiv \QQ(d) = \{q_{ij}\}_{i,j=1,\hdots,d} \in \mathcal{R}(d)$, $d \in \bbN$,
  where either
  \begin{align}\label{e:RM_theorem_general_condition}
    \begin{aligned}
  &\{q_{ij}\}_{i , j =1, \hdots, d, \hspace{0.75mm}i \neq j}
    \stackrel{\textnormal{iid}}\sim F_X
    \text{ and }\\
  &- q_{ii} = \sum_{j \in \{1,\hdots,d\} \backslash \{i\}}q_{ij},
  i=1,\hdots,d
  \end{aligned}
  \end{align}
  or we impose that $\QQ \in \mathcal{R}(d) \cap \mathcal{S}(d)$ as 
    \begin{align}
    \begin{aligned}\label{e:RM_theorem_symmetry_condition}
  &\{q_{ij}\}_{i , j =1, \hdots, d, \hspace{0.75mm}i > j}
  \stackrel{\textnormal{iid}}\sim F_X,
  q_{ij} := q_{ji} \text{ for } i < j
    \text{ and }\\
  &- q_{ii} = \sum_{j \in \{1,\hdots,d\} \backslash \{i\}}q_{ij},
  i=1,\hdots,d.
  \end{aligned}
\end{align}
    Then, in either of these cases, for any $d \in \bbN \backslash\{1\}$, we have
    \begin{align}
 \mu+ O_{\as}\Big(\sqrt{ \frac{\log d}{d}}\Big) \leq& \frac{\sigma_2(\QQ)}{d} 
 \notag \\
 \leq& \frac{\sigma_d(\QQ)}{d}
 \leq \mu+ O_{\as}\Big(\sqrt{ \frac{\log d}{d}}\Big)
 \label{e:lambda_n-1(Q)->-infty_non-symm}
      \end{align}
      almost surely. In \eqref{e:lambda_n-1(Q)->-infty_non-symm},
      $O_{\textnormal{a.s.}}$ is as in
      \eqref{e:X=O_a.s.(1)_Y=o_a.s.(1)}.
\end{theorem}

\begin{remark}
  The bounds constructed in Theorem
  \ref{t:lambda_n-1(Q)->-infty_non-symm_a.s.} are strongly reminiscent
  of the bounds for eigenvalues provided in Theorem 1.5 of the seminal
  paper \cite{bordenave:caputo:chafai:2014} (see also, for instance,
  Corollary 1.6 in \cite{bryc:dembo:jiang:2006} and Corollary 1.1 in
  \cite{ding:jiang:2010}).
\end{remark}

Finally, let us observe that \cref{thm:bride:of:frankinstein},
\cref{t:lambda_n-1(Q)->-infty_non-symm_a.s.} as well as the fact
that
\begin{align}
  \sigma_j(\QQ) = |\lambda_{d - j +1}(\QQ)|,
  \text{ for e.g. any }
  \QQ \in \mathcal{S}(d) \cap \mathcal{R}(d)
  \label{eq:eigen:to:SV:SYM:dummy}
\end{align}
combine to produce the following immediate corollary.
\begin{Corollary}\label{cor:rig:det:rand:wrapup}
  Consider any sequence of random matrices
  $\QQ \equiv \QQ(d) = \{q_{ij}\}_{i,j=1,\hdots,d} \in
  \mathcal{R}(d)$, $d \in \bbN$, as in
  \cref{t:lambda_n-1(Q)->-infty_non-symm_a.s.} under the second
  (symmetric) case \eqref{e:RM_theorem_symmetry_condition}.  Then,
  taking $\mu_1 = \mu_1(d)$ and $\mu_2= \mu_2(d)$ as the resulting
  lower and upper bounds defined by
  \eqref{e:lambda_n-1(Q)->-infty_non-symm}, we have that $\QQ$
  satisfies both \eqref{eq:bride:of:frankinstein:1:rm:app},
  \eqref{eq:bride:of:frankinstein:0:rm:app}, cf.
  \eqref{eq:eigen:to:SV:SYM:dummy} relative to this sequence of
  $\mu_1, \mu_2$ for any $\JJ \in \mathcal{M}(d)$.
\end{Corollary}

\begin{figure*}[!t]
    \centering
    \includegraphics[width=\linewidth]{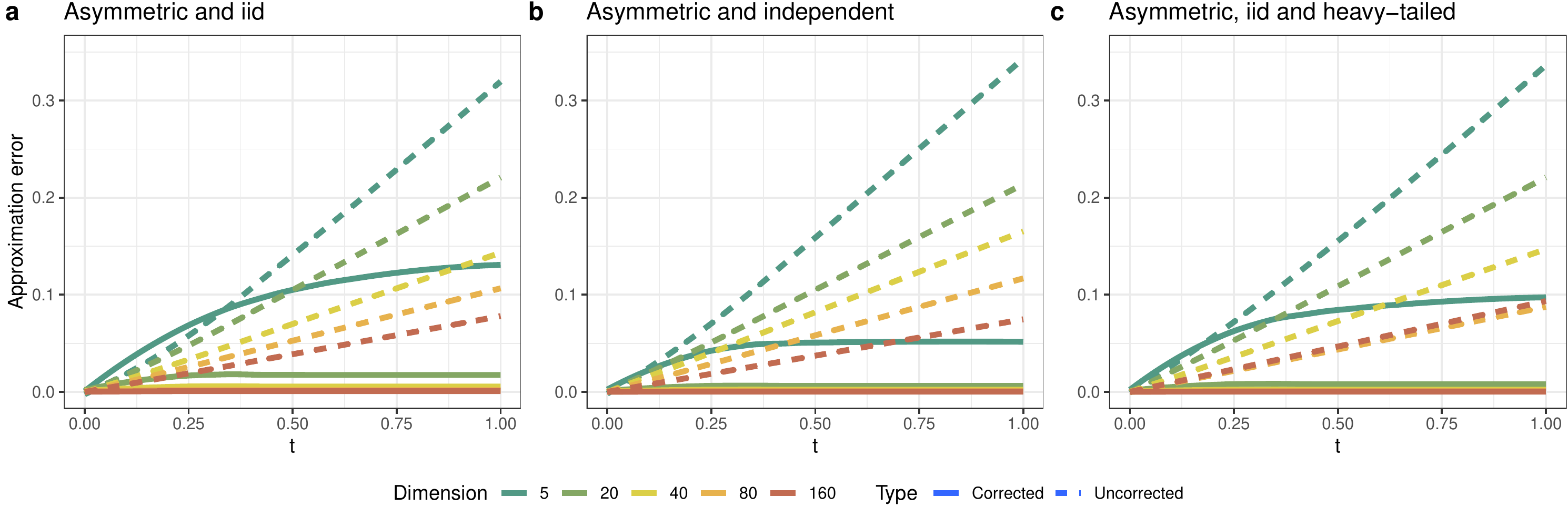}
    \caption{Frobenius norm errors obtained by first-order approximation $te^{t\QQ}\JJ$ and by affine-corrected first-order approximation
      $te^{t\QQ}\JJ - t(\II -\QQ^+ \QQ) \JJ \QQ \QQ^{+}$
      (\cref{thm:bride:of:frankinstein}) under increasingly relaxed
      assumptions.  Within each assumption set, we average over 20
      independent Monte Carlo simulations of random generator matrices
      for each dimension. Plot \textbf{a} corresponds to asymmetric
      generators with off-diagonal elements having independent and
      identically distributed (iid) standard exponential random
      variables that correspond to the sub-exponential distribution
      hypothesis. Plot \textbf{b} drops the identical distribution
      assumption by allowing each row and column of the generator
      matrix to additively contribute its own mean---itself given by a
      standard exponential---to its corresponding exponentially
      distributed entries. Plot \textbf{c} features rate
      matrices with iid Cauchy entries truncated to be
      positive. Empirically, the results of
      \cref{cor:rig:det:rand:wrapup} extend beyond the symmetric,
      iid and sub-exponential hypotheses, suggesting scope of future work.}
    \label{fig:app}
\end{figure*}

\begin{remark}\label{rmk:non:sym:num:ev}
  Our rigorous formulation of \cref{cor:rig:det:rand:wrapup} is
  limited to symmetric random rate matrices whose above diagonal
  elements are independent and identically distributed draws from a
  sub-exponential distribution. However, strong numerical evidence
  suggest that the scope of the approximations
  \eqref{eq:bride:of:frankinstein:1:rm:app},
  \eqref{eq:bride:of:frankinstein:0:rm:app} reach far beyond the
  limitations of \cref{cor:rig:det:rand:wrapup} in several different
  ways. \cref{fig:app} explores the
  consequences of relaxing various assumptions of
  \cref{cor:rig:det:rand:wrapup}.  The third plot involves folded Cauchy random variables, the heavy tails of which violate the sub-exponential assumption (\cref{def:sub:exp:RV}).  There appears to be no
  significant departure from the idea that the form
  $t(\II -\QQ^+ \QQ) \JJ \QQ \QQ^{+}$ corresponds increasingly well to
  the true approximation error (\eqref{eq:true:diff}) as the dimension
  increases.
\end{remark}

\begin{remark}\label{rmk:better:approx}
  Calculation of
  $\widetilde{\nabla}_{\JJ} e^{t\QQ} := t e^{t\QQ} \JJ$ requires
  $\order{d^3}$ operations by, e.g., computing the spectral
  decomposition $\QQ$ as in \eqref{eq:spec:decomp}.  One may then recycle this
  decomposition to determine the
  additional term $t(\II -\QQ^+ \QQ) \JJ \QQ \QQ^{+}$ for little extra cost.  In view of
  \cref{cor:rig:det:rand:wrapup} and \cref{fig:app},
  \begin{align}\label{eq:better:help:on:sale}
    \widehat{\nabla}_{\JJ} e^{t\QQ} := t e^{t\QQ} \JJ
    -t(\II -\QQ^+ \QQ) \JJ \QQ \QQ^{+}
  \end{align}
  provides an accurate approximation of $\nabla_{\JJ} e^{t\QQ}$ for
  asymptotically for large $d$. Thus, we anticipate further computational improvements when fitting
  large, gold-standard models using this
  refined approximation.  That said, we leave the efficient and scalable application of $\widehat{\nabla}_{\JJ} e^{t\QQ}$ to future work.
\end{remark}

\begin{figure}[!t]
	\centering
	\includegraphics[width=\linewidth]{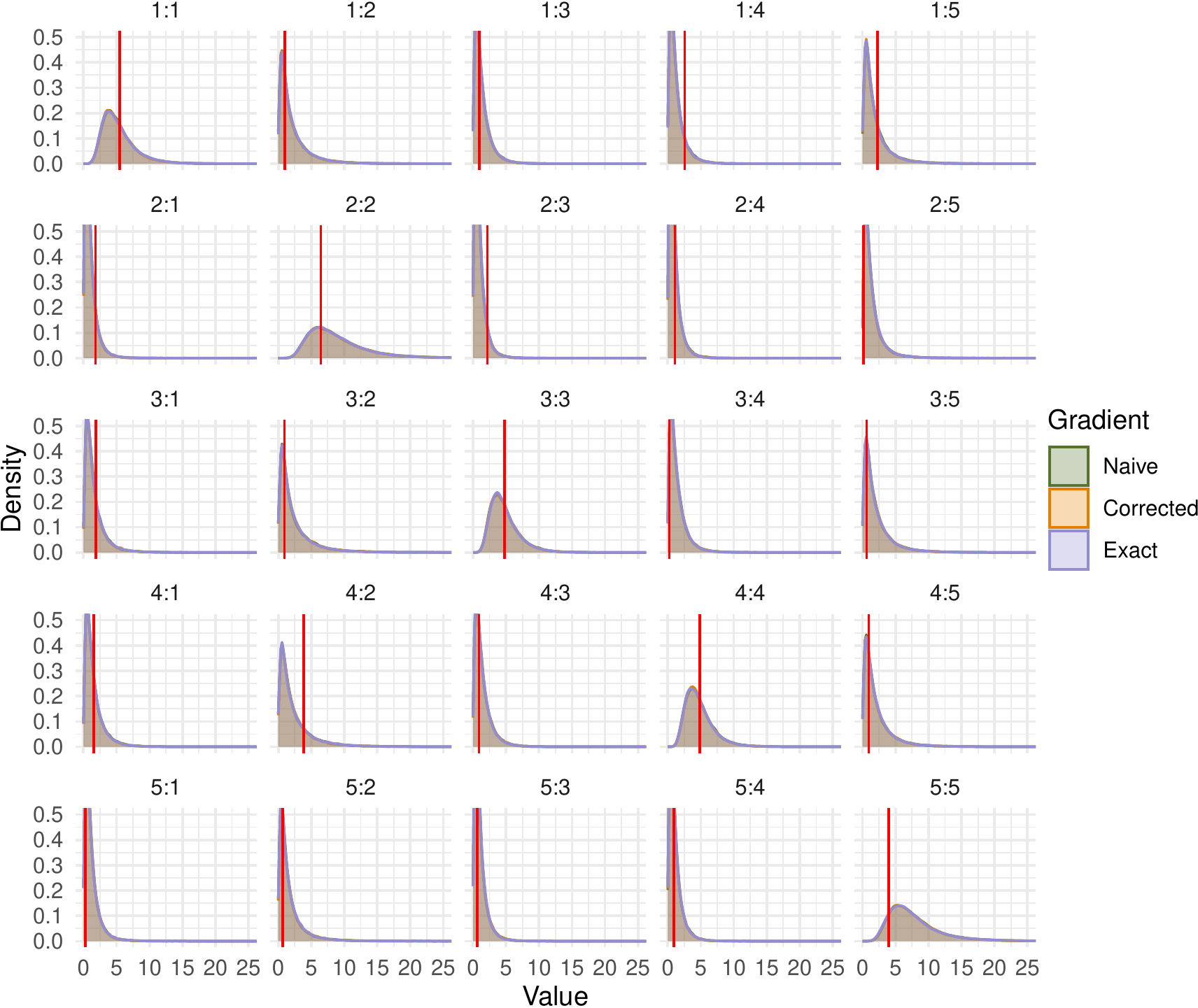}
	\vspace{-2em}
	\caption{Posterior density plots for surrogate-trajectory Hamiltonian Monte Carlo (HMC) (\cref{sec:Sur:trad:HMC}) using the naive approximate derivative, the corrected approximation \eqref{eq:better:help:on:sale} and the exact matrix exponential derivative for elements of a $5\times 5$ generator matrix. The near perfect overlap reflects the fact that each algorithm's transition kernel leaves the posterior distribution invariant \cite{glatt2020accept}.  To generate data, we randomly draw standard normal generator entries once and simulate 20 independent initial/final position pairs from a CTMC with time interval $t=1$.  We show the true value in red and negate diagonal elements to simplify presentation.  }\label{fig:lowDim}
\end{figure}

\section{Empirical Studies}\label{sec:emp}

\begin{figure*}[!t]
	\centering
	\includegraphics[width=\linewidth]{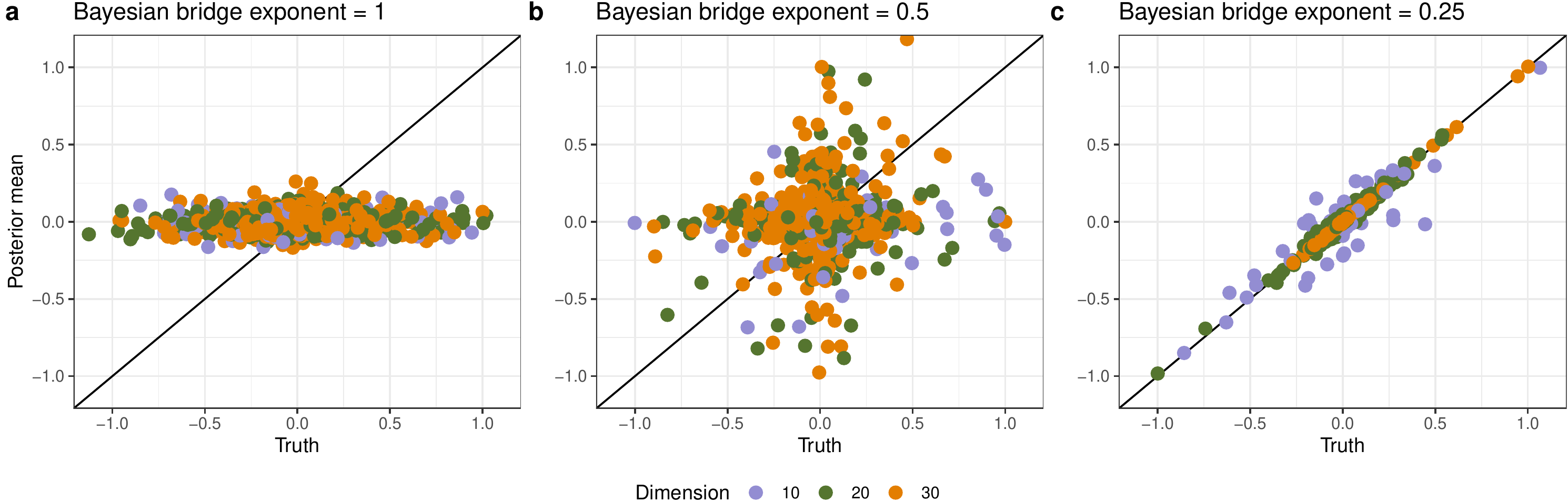}
	\caption{Posterior means versus truth for CTMC generator matrix elements within differing sparsity regimes and with different dimensionalities $d$, holding observation count fixed at 300. We generate posteriors using surrogate-trajectory Hamiltonian Monte Carlo with the naive matrix exponential derivative.  To affect sparsity levels, we generate generator matrix entries according to the Bayesian bridge distribution \cite{polson2014bayesian} with different exponents ($\alpha \in \{1,1/2,1/4\}$), normalizing by the largest absolute values to ease comparison. Smaller $\alpha$ values encode more peaked distributions with heavier tails and thus enforce greater sparsity.   Plot \textbf{c} reflects the fact that the Bayesian bridge prior with exponent $\alpha=1/4$ helps identify non-null parameters in small sample contexts \cite{magee2023random}.  With this intuition in mind, we specify such a prior on generator random effects in \cref{sec:sars:cov:app}.}
	\label{fig:truthMean}
\end{figure*}

Before applying the naive matrix exponential derivative approximation to the phylogeographic analysis of SARS-CoV-2, we carry out a few targeted studies that illustrate: the empirical performance of the approximate derivative and its affine correction \eqref{eq:better:help:on:sale} (\cref{fig:app}); agreement between CTMC generator matrix empirical posterior distributions generated by surrogate-trajectory Hamiltonian Monte Carlo (HMC) algorithms using approximate derivatives and the truth (\cref{fig:lowDim}); and point estimation of generator matrix element values under different sparsity regimes and different CTMC state space dimensionalities $d$ (\cref{fig:truthMean}). 

Here, we fill in remaining simulation details not included in figure captions. In the simulations contributing to \cref{fig:app}, we randomly generate new, independent direction matrices $\JJ$ at each time step within each of the 20 independent runs. The results are not sensitive to $\JJ$ in general.  We also note that under \cref{def:sub:exp:RV} the Cauchy distribution is \emph{not} sub-exponential and therefore represents a deviation from core assumptions of \cref{sec:rig:results}.  Using varying distributions on the elements of generator matrices $\QQ$, the simulations contributing to \cref{fig:lowDim} and \cref{fig:truthMean} both randomly generate $N$ independent initial states $\y^0_1,\dots,\y^0_N$ according to uniform distributions over their respective CTMC state spaces. For each of these initial states $\y_0$, we then simulate from the CTMC for time interval $t=1$ to obtain samples $\y^1_1,\dots,\y^1_N$.  The likelihood then takes the form $\prod_n (\y^{0}_n)^T e^{\QQ} \y^{1}_n$.  For the simulation contributing to \cref{fig:lowDim}, we specify standard normal priors and sample from the posterior by generating 100,000 iterations from each algorithm.  For the simulation contributing to \cref{fig:truthMean}, we generate 500,000 MCMC samples using surrogate-trajectory HMC with the naive matrix exponential derivative and calculate the posterior mean of each generator matrix element using the final 100,000 samples.

\section{Application: Global Spread of SARS-CoV-2}
\label{sec:sars:cov:app}

\cite{magee2023random} consider phylogenetic models that involve CTMC
priors and show that the first-order approximation
\eqref{eq:firstOrder} of the matrix exponential derivative
\eqref{eq:grad:def} performs remarkably well within surrogate HMC (see
\cref{sec:Sur:trad:HMC} in the Supplement), achieving an over 30-fold
efficiency gain compared to random-walk Metropolis for a 14-state
model with over 180 model parameters.  Here, we demonstrate similar
strong performance of the first-order approximation within surrogate
HMC for a 44-state model with almost 1,900 model parameters.

\subsection*{Phylogenetic CTMC}

Start with a (possibly unknown) rooted and bifurcating phylogenetic
tree $\phylogeny$ consisting of $N$ leaf nodes that correspond to
observed biological specimens and $N-1$ internal nodes that correspond
to unobserved ancestors.  The tree also contains $2N-2$ branches of
length $t_v$ connecting each child node $v$ to its parent $u$.

Given $\phylogeny$, we model the evolution of $d$ characters along
each branch of the tree according to a CTMC model with $d\times d$
generator matrix $\QQ$ and stationary distribution
$\ppi=(\pi_1,\dots,\pi_d)=\lim_{t\rightarrow
  \infty}\widetilde{\ppi}e^{t\QQ}$, for $\widetilde{\ppi}$ any
arbitrary probability vector.  Examples of characters are the $d=4$
nucleotides within a set of aligned genome sequences
\cite{jukes1969evolution} and the set of $d=15$ geographic regions
visited by an influenza subtype \cite{lemey2014unifying}.  We may
scale $t_v$ to be raw time (e.g., years) or the expected number of
substitutions with respect to $\ppi$ depending on the given problem.
In the former case, one may augment the model with a rate scalar
$\gamma$ that modulates the expected number of substitutions across
all branches, and the finite-time transition probability matrix along
branch $v$ becomes $\P_v:=e^{\gamma t_v \QQ}$.  In the following, we
further posit $\QQ=\QQ(\ttheta)$ for $\ttheta$ a vector of parameters.
 
Let data $\Y$ be the $d\times N$ matrix with columns $\y_n$, $n \in \{1,\dots,N\}$, each
having a single non-zero entry (set to 1) corresponding to the
observed state of the biological specimen.  One may use
any node $v$ to express the likelihood
 \begin{align}
p(\Y|\ttheta) = \p_v^T \q_v \, ,
 \end{align}
 where $\p_v$ and $\q_v$ are the post-order and pre-order partial
 likelihood vectors, respectively \cite{ji2020gradients}.  The former describes the
 probability of the observed states for all observed specimens (i.e.,
 leaf nodes) that descend from node $v$, conditioned on the state at
 node $v$. The latter describes analogous probabilities for all
 observed specimens not descending from node $v$. For leaf nodes,
 $\p_n:=\y_n$, $n \in \{1,\dots,N\}$, and one may specify the root node's pre-order partial likelihood to be any arbitrary probability vector \emph{a priori}. Let `$\circ$' denote the Hadamard or elementwise product between matrices or vectors of equal dimensions.  If we suppose that node $u$ gives rise to two child
 nodes, $v$ and $w$, then
\begin{align}\label{eq:recursions}
  \p_u = \P_v \p_v \circ \P_w \p_w\, ,
  \quad \q_v= \P_v^T \left( \q_u \circ  \P_w \p_w\right) \, .
\end{align}
Using the chain rule and the fact that $\p_v$ does not depend on
$\P_v$, one may write the likelihood's derivative with respect to a
single parameter $\theta_k$ thus:
\begin{align}\label{eq:phyloGrad}
	\frac{\partial}{\partial \theta_k}  p(\Y|\ttheta) 
  &\propto   \sum_{v=1}^{2N-2}
    \mbox{tr} \left( \frac{\partial (\p_v^T\q_v)}{\partial \P_v}
    \frac{\partial \P^T_v}{\partial \theta_k} \right) \\ \notag
  &= \sum_{v=1}^{2N-2} \mbox{tr} \left(  (\q_u\circ \P_w\p_w)
    \p_v^T\bigg( \frac{\partial e^{\gamma t_v \QQ}}{\partial \theta_k}\bigg)^{\!T}\right)  \\ \nonumber
  &=\sum_{v=1}^{2N-2} \p_v^T
    \Biggl(\sum_{i,j=1}^d  \frac{\partial e^{\gamma t_v \QQ}}{\partial q_{ij}} \frac{\partial q_{ij}}{\partial \theta_k} \Biggr)^{\!T}
    \hspace{-0.3em}(\q_u\circ \P_w\p_w)  \, ,
\end{align}
where we suppress the dependence of $u$ and $w$ on $v$.

Whereas the recursions of \eqref{eq:recursions} facilitate fast
likelihood computation, inferring $\ttheta$ using gradient-based
approaches such as HMC requires a large number of repeated evaluations
of the matrix exponential derivative.  These computations become
particularly onerous when one opts for a gold-standard mixed-effects
model \cite{magee2023random} and specifies
\begin{align}\label{eq:mixed}
\log q_{ij} = b_{ij} + \epsilon_{ij} \,, \quad i\neq j\, .
\end{align} 
Here, the fixed effects $b_{ij}$ are elements of some non-random
matrix $\B$, and the random effects $\epsilon_{ij}$ are mutually
independent \emph{a priori} and inferred as model parameters.  The
dimension $K$ of $\ttheta$ in this model is $\mathcal{O}(d^2)$, so the $\mathcal{O}(KNd^5)$ cost of the log-likelihood derivative \eqref{eq:phyloGrad}
becomes a massive $\mathcal{O}(Nd^7)$.  In this context,
\cite{magee2023random} show that an approximate log-likelihood derivative based on the first-order approximation to the matrix exponential helps achieve a considerable speedup over the exact
derivative, requiring only $\mathcal{O}(d^4 + Nd^3)$ floating-point operations yet facilitating high-quality proposals in the context
of surrogate HMC.

In the following section, we use this method to analyze the global
spread of SARS-CoV-2 and show that the first-order approximation
maintains its performance for an even higher-dimensional problem than
previously considered.

\subsection*{Bayesian analysis of SARS-CoV-2 contagion}

\begin{figure*}[!t]
	\centering
	\includegraphics[width=\linewidth]{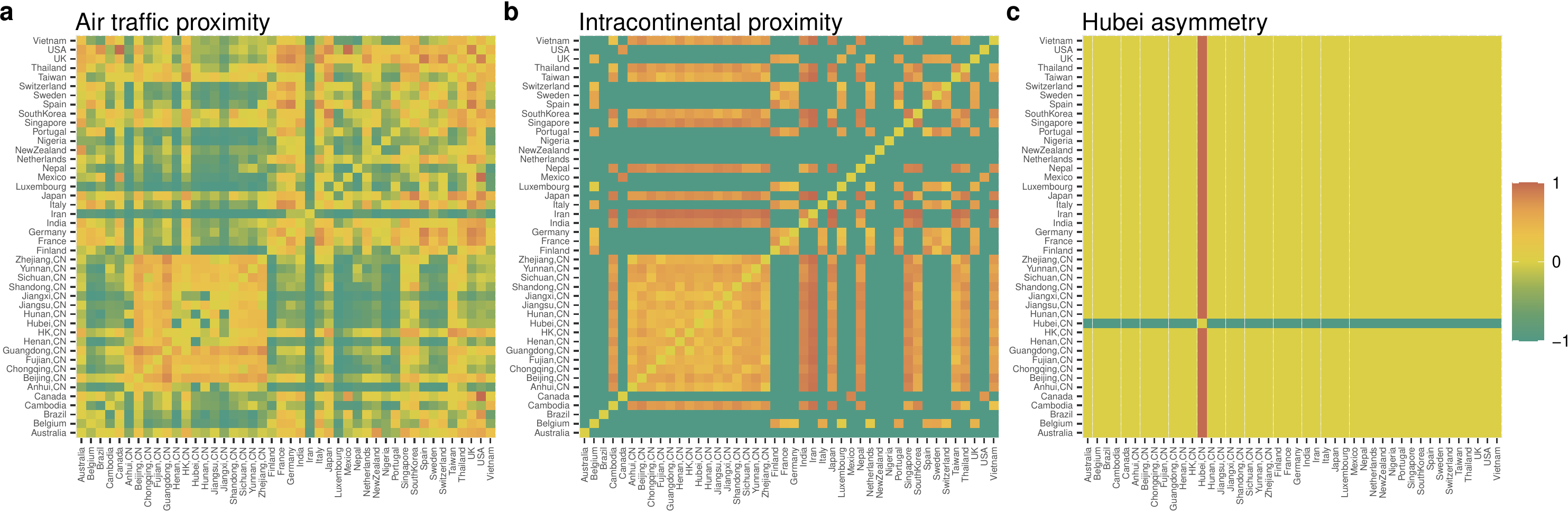}
	\vspace{-2em}
	\caption{Matrix predictors described in \eqref{eq:linear} combine in a linear manner to form the fixed-effect matrix $\B$ featured within the mixed-effects regression model \eqref{eq:mixed}.  Air traffic proximities \emph{(\textbf{a})} are proportional the number of air passengers exchanged between airports within respective regions \cite{holbrook2021massive}.  Intracontinental proximities \emph{(\textbf{b})} take values between -1 for regions on different continents to 1 for adjacent regions.  The Hubei asymmetry \emph{(\textbf{c})} roughly characterizes the Hubei quarantine of early 2020.}\label{fig:FE}
\end{figure*}

We use the phylogenetic CTMC framework to model the early spread of
SARS-CoV-2---the virus responsible for the ongoing COVID-19
pandemic---based on $N=285$ observed viral samples collected from $31$
regions worldwide between December 24, 2019 and March 19, 2020. These
regions comprise 13 provinces within China and 18 countries without.
Understanding the manner in which viruses travel between human
populations is an object of ongoing study, and phylogeographic
analyses point to the central role of travel networks including those
measured by airline passenger counts \cite{holbrook2021massive} or
google mobility data \cite{worobey2020emergence}.  Here, we include
three such predictors of travel in our CTMC model by expanding the
fixed effects in regression model \eqref{eq:mixed} to take the form
\begin{align}\label{eq:linear}
	\B = \theta_1 \X_1 + \theta_2  \X_2 + \theta_3 \X_3 
\end{align} 
for $\X_1$, $\X_2$ and $\X_3$ fixed $44\times 44$ matrix predictors and $\theta_1$, $\theta_2$ and $\theta_3$
 real-valued regression coefficients. Note we have expanded the
number of regions between which viruses may travel to $d=44$ by
including two additional Chinese provinces and 11 additional
countries.  Figure \ref{fig:FE} presents the three predictors of
interest: $\X_1$ contains air travel proximities between locations as
measured by annualized air passenger counts between airports contained
within a region \cite{holbrook2021massive}; $\X_2$ contains
intracontinental proximities arising from physical distances when two
regions inhabit the same continent and fixed at the minimum otherwise;
finally, $\X_3$ describes the Hubei asymmetry, i.e., the non-existence
of human travel out of the Hubei province in early 2020.

\begin{figure*}[!t]
	\centering
	\includegraphics[width=\linewidth]{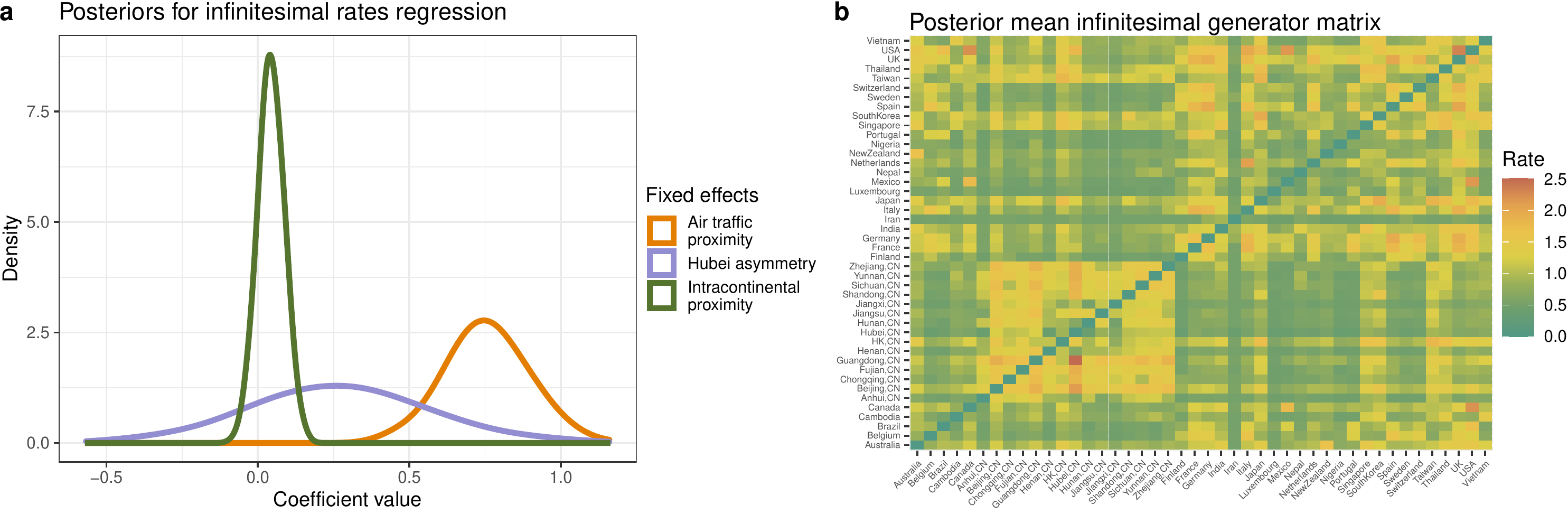}
	\vspace{-2em}
	\caption{Posterior inference. Posterior densities
          \emph{(\textbf{a})} for the three fixed-effect regression
          coefficients corresponding to the predictor matrices of
          Figure \ref{fig:FE} reflect largely positive associations
          between predictors and infinitesimal rate matrix, although
          air traffic proximity has the only statistically significant
          coefficient with posterior mean of 0.76 and 95\% credible
          interval of (0.50, 1.03).  The posterior mean
          \emph{(\textbf{b})} infinitesimal rate matrix closely
          resembles the air traffic predictor, while reflecting the
          Hubei asymmetry to a lesser extent.}\label{fig:post}
\end{figure*}

In the context of a Bayesian analysis, we specify independent normal
priors on $\theta_1$, $\theta_2$ and $\theta_3$ with means of 0 and
variances of 2.  We also assume that the 1,892 random effects
$\epsilon_{ij}$ follow sparsity inducing Bayesian bridge priors with
global scale parameter $\tau$ and exponent $\alpha=0.25$.  Here, we
follow \cite{nishimura2022shrinkage} and specify a Gamma prior on
$\tau^{-\alpha}$ with a shape parameter of 1 and a rate parameter of
2. Finally, we place a flat prior on the rate scalar $\gamma$.
Inferring the posterior distribution of all 1,897 model parameters
requires an advanced MCMC strategy. Namely, we adopt an
HMC-within-Gibbs approach, updating the scalars $\gamma$ and $\tau$
independently but updating all 1,895 regression parameters (both fixed
and random effects) using surrogate HMC accomplished with the
first-order approximation
\begin{align*}
  \gamma t_v  e^{\gamma t_v \QQ} \JJ_{ij}
  \approx \frac{\partial e^{\gamma t_v \QQ}}{\partial q_{ij}} 
\end{align*}
within \eqref{eq:phyloGrad}.

We generate 8 million MCMC samples in this manner, saving 1 in 10,000
Markov chain states, in order to guarantee a minimum effective sample
size (ESS) greater than 100.  By far, the worst-mixing parameter is
the global scale $\tau$, which obtains an ESS of 185.  The three
fixed-effects regression coefficients $\theta_1$, $\theta_2$ and
$\theta_3$ obtain ESS's of 721, 721 and 644, respectively.  The median
and minimum ESS for the 1,892 random effects $\epsilon_{ss'}$ are 721
and 190, respectively.  We note that, after thinning and removing
burn-in, the sample only consists of 721 states, so ESS of 721 implies
negligible autocorrelation between samples.

The left plot of Figure \ref{fig:post} presents posterior densities
for the fixed-effect coefficients from regression \eqref{eq:mixed}.
The air traffic proximity, intracontinental proximity and Hubei
asymmetry coefficients have posterior means and 95\% credible
intervals of 0.76 (0.50, 1.03), 0.04 (-0.04, 0.12) and 0.27 (-0.39,
0.78), respectively. While these results suggest a positive
association between each predictor matrix and the infinitesimal
generator matrix $\QQ$, the only predictor with a statistically
significant association is air traffic proximity.  This result agrees
with a previous phylogeographic analysis of the global spread of
influenza \cite{holbrook2021massive}.  The right plot of Figure
\ref{fig:post} presents the posterior mean for generator $\QQ$.  As
one may expect, the matrix looks similar to that of the air traffic
predictor matrix in Figure \ref{fig:FE}, but one may also see the
influence of the Hubei asymmetry in, e.g., the squares corresponding
to travel between Guangdong and Hubei provinces.  Finally, we randomly
generate regions of unobserved ancestors from their posterior
predictive distributions every 100-thousandth MCMC iteration.
After collapsing regions into 8 major blocks, Figure \ref{fig:tree} projects the empirical posterior predictive mode of these blocks onto the phylogenetic tree $\phylogeny$.  The general pattern looks similar to that of Figure 1 from \cite{lemey2020accommodating}, although the geographic blocking scheme differs slightly.

\begin{figure*}[!t]
	\centering
	\includegraphics[width=0.7\textwidth]{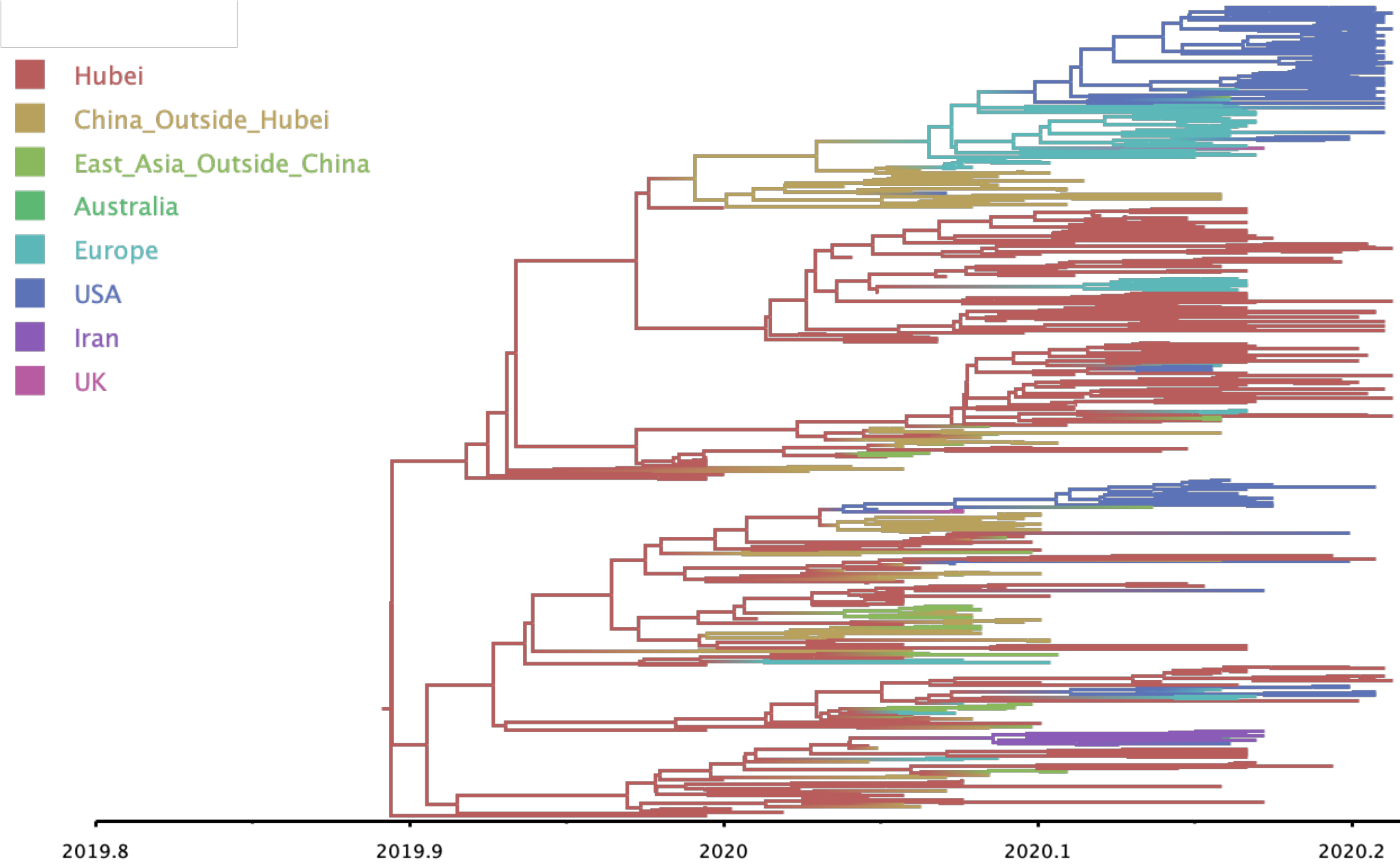}
	\vspace{-0.5em}
	\caption{Posterior predictive modes for (unobserved) ancestral
          locations color a phylogenetic tree that describes the
          shared evolutionary history of 285 SARS-CoV-2
          samples.}\label{fig:tree}
\end{figure*}

In addition to these scientific questions of interest, we are
interested in the performance of the first-order approximation as a
surrogate gradient for HMC in such a high-dimensional setting. Whereas we know that the surrogate-trajectory HMC transition kernel leaves its target distribution invariant regardless of the approximation quality \cite{glatt2020accept}, transitions that rely on poor gradient approximations result in small acceptance rates, more random walk behavior and high autocorrelation between samples.  Since ESS is inversely proportional to a Markov chain's asymptotic autocorrelation, larger ESS suggest a useful gradient approximation.  To
isolate the approximation's performance, we fix the Bayesian bridge global-scale
$\tau$ at 2.5$\times10^{-5}$.  We generate a Markov chain with 80,000
states, saving every tenth state and removing the first 1,000 states as
burn-in.  Despite the relatively small number of iterations, we
observe large ESS that suggest satisfactory accuracy of the
first-order approximation within the context of high-dimensional
surrogate HMC.   The ESS for the three fixed-effect regression
coefficients $\theta_1$, $\theta_2$ and $\theta_3$ are 1,053, 1,343
and 498, respectively.  The median and minimum ESS for the 1,892
random effects $\epsilon_{ij}$ are 1,514 and 1,161, respectively.

\section*{Discussion}

We develop tight probabilistic bounds on the error associated
with a simplistic approximation to the matrix exponential derivative
for a large class of CTMC infinitesimal generator matrices with random
entries.  Our ``blessing of dimensionality'' result shows that this
error improves for higher-dimensional matrices.  We apply the
numerically naive approach to the analysis of the global spread of
SARS-CoV-2 using a mixed-effects model of unprecedented dimensions.
The results obtained herein suggest the further study of CTMCs through
the lens of random matrix theory.  Furthermore, this analysis suggests a refinement of the first-order
approximation to the matrix exponential derivative that may be particularly useful within
modern, high-dimensional settings.

\acknow{GD was partially supported by the Simons Foundation
  collaboration grant $\#$714014. NEGH received support for this work
  under NSF DMS 2108790.  AJH is supported by a gift from the Karen
  Toffler Charitable Trust and by grants NIH K25 AI153816, NSF DMS
  2152774 and NSF DMS 2236854. AFM and MAS are partially supported
  through grants NIH R01AI153044 and R01AI162611.}

\showacknow{} 

\bibliography{refs}

\newpage


\appendix

\section{Proof of \cref{thm:ass:error}}
\label{sec:Thm:det:main}

Recall that $\Z(t)$ is defined as in \eqref{eq:true:diff}, and
consider the representation \eqref{eq:ass:procs:0}. We begin by
decomposing the first term $e^{-s\QQ} \JJ e^{s \QQ}$ inside of the
integral in \eqref{eq:ass:procs:0} as
\begin{align}
  e^{-s\QQ} &\JJ e^{s \QQ} \notag\\
  =&e^{-s\QQ} \QQ^+ \QQ \JJ \QQ \QQ^{+} e^{s\QQ}
     + e^{-s\QQ} (\II -\QQ^+ \QQ) \JJ \QQ \QQ^{+} e^{s\QQ}
                  \notag\\
      &+ e^{-s\QQ}\QQ^+ \QQ \JJ (\II - \QQ \QQ^{+}) e^{s\QQ}
                \notag\\
       &+e^{-s\QQ} (\II -\QQ^+ \QQ) \JJ (\II -\QQ \QQ^{+}) e^{s\QQ}
                           \notag\\
  = &e^{-s\QQ} \QQ^+ \QQ \JJ \QQ \QQ^{+} e^{s\QQ}
      + (\II -\QQ^+ \QQ) \JJ \QQ \QQ^{+} e^{s\QQ}
         \label{eq:big:yuck:1}                          
      \\
      &+ e^{-s\QQ}\QQ^+ \QQ \JJ (\II - \QQ \QQ^{+})
  + (\II -\QQ^+ \QQ) \JJ (\II -\QQ \QQ^{+}),
  \notag
\end{align}
where we used \eqref{eq:MP:proj:prop} for the second identity.
Regarding the middle two terms, with our commutativity
\eqref{eq:GI:com} assumption, we observe that
\begin{align}
  (\II &-\QQ^+ \QQ) \JJ \QQ \QQ^{+} e^{s\QQ}
         + e^{-s\QQ}\QQ^+ \QQ \JJ (\II - \QQ \QQ^{+})
                \label{eq:big:yuck:2}                          
                     \\
  &=  (\II -\QQ^+ \QQ) \JJ \QQ^{+} \QQ e^{s\QQ}
    + e^{-s\QQ}\QQ \QQ^+ \JJ (\II - \QQ \QQ^{+})
    \notag\\
  &= \frac{d}{ds}
    \left(
    (\II -\QQ^+ \QQ) \JJ \QQ^{+} e^{s\QQ}
    - e^{-s\QQ} \QQ^+ \JJ (\II - \QQ \QQ^{+})
    \right).
 \notag
\end{align}
Hence, integrating in time, and then using \eqref{eq:GI:prop},
\eqref{eq:GI:com} for the first and last term terms on the right hand side, we have
\begin{align}
  \int_0^t& \big\{(\II -\QQ^+ \QQ) \JJ \QQ \QQ^{+} e^{s\QQ}
                             + e^{-s\QQ}\QQ^+ \QQ \JJ (\II - \QQ \QQ^{+})\big\}ds
                             \notag\\
   =& (\II -\QQ^+ \QQ) \JJ \QQ^{+}e^{t\QQ}
      - e^{-t\QQ} \QQ^+ \JJ (\II - \QQ \QQ^{+})
      \notag\\
    &- (\II -\QQ^+ \QQ) \JJ \QQ^{+} 
      +  \QQ^+ \JJ (\II - \QQ \QQ^{+})
      \notag\\
   =& (\II -\QQ^+ \QQ) \JJ (\QQ^+)^2 \QQ e^{t\QQ}
      - e^{-t\QQ} \QQ^+ \JJ (\II - \QQ \QQ^{+})
      \notag\\
    &- (\II -\QQ^+ \QQ) \JJ \QQ^{+} 
      + \QQ (\QQ^+)^2 \JJ (\II - \QQ \QQ^{+})
             \label{eq:big:yuck:3}                          
\end{align}
Finally let us note, regarding the second term on the right hand side
of \eqref{eq:ass:procs:0}, using again 
\eqref{eq:MP:proj:prop} and \eqref{eq:GI:prop}, \eqref{eq:GI:com}
\begin{align}
  -te^{t\QQ}\JJ 
  =-te^{t\QQ}\QQ \QQ^+ \JJ  - t (\II - \QQ^+\QQ) \JJ.
               \label{eq:small:yuck:4}                          
\end{align}

By combining \eqref{eq:big:yuck:1}--\eqref{eq:small:yuck:4} and
comparing with \eqref{eq:ass:procs:0}, we now find that
\begin{align}\label{e:E(t)=integ+sum}
  \Z(t)\! =\! \! \int_0^t \! \! \! e^{(t-s)\QQ} &\QQ^+ \QQ \JJ \QQ \QQ^{+} e^{s\QQ}ds
           + (\II - \!\QQ^+ \QQ) \JJ (\QQ^{+})^2\QQ e^{t\QQ}
           \notag\\
           &+ e^{t\QQ} \QQ (\QQ^+)^2 \JJ (\II - \QQ \QQ^{+}) -te^{t\QQ}\QQ \QQ^+ \JJ 
           \notag\\
           &-\QQ^+ \JJ (\II - \QQ \QQ^{+})- (\II -\QQ^+ \QQ) \JJ \QQ^{+}
             \notag\\
         &-t(\II -\QQ^+ \QQ) \JJ \QQ \QQ^{+}.
\end{align}
Rearranging \eqref{e:E(t)=integ+sum}, we have 
\begin{align*}
\QQ^+ \JJ (\II - \QQ \QQ^{+}) &+(\II -\QQ^+ \QQ) \JJ \QQ^{+}
   +t(\II -\QQ^+ \QQ) \JJ \QQ \QQ^{+}
                                \notag\\
  &+  \nabla_\JJ e^{t \QQ} - t  e^{t\QQ} \JJ
  = \TT_1 + \TT_2,
\end{align*}
where 
\begin{align*}
 \TT_1 :=&  \int_0^t  \! \! e^{(t-s)\QQ} \QQ^+ \QQ \JJ \QQ \QQ^{+} e^{s\QQ}ds,\\
  \TT_2 :=&  (\II -\QQ^+ \QQ) \JJ (\QQ^{+})^2\QQ e^{t\QQ}\\
         &+ e^{t\QQ} \QQ (\QQ^+)^2 \JJ (\II - \QQ \QQ^{+}) -te^{t\QQ}\QQ \QQ^+ \JJ.
\end{align*}

We estimate each of $\TT_1$ and $\TT_2$ in turn.  Regarding $\TT_1$,
using once more \eqref{eq:GI:com}, as well as \eqref{eq:qetq:exp:decay}, we have
\begin{align*}
  \| \TT_1\| \leq&
    \int_0^t \|\QQ e^{(t-s)\QQ}\| \| \QQ^+  \JJ \QQ^{+} \|\| \QQ e^{s\QQ}\|ds
         \notag\\
   \leq& C_0^2\int_0^t  e^{-(t-s)\kappa} \| \QQ^+  \JJ \QQ^{+} \| e^{-s\kappa } ds
                         \notag\\
   =& C_0^2 t e^{-t\kappa} \| \QQ^+  \JJ  \QQ^{+} \|.
\end{align*}
Turning to $\TT_2$, by \eqref{eq:qetq:exp:decay} we have
\begin{align*}
  \| \TT_2\|
  \leq C_0 (&\| (\II -\QQ^+ \QQ) \JJ (\QQ^{+})^2\|
               + \|(\QQ^+)^2\JJ (\II - \QQ \QQ^{+})\|
               \notag\\
               &+ t\|\QQ^+ \hspace{0.25mm}\JJ\|) e^{- t \kappa }.
\end{align*}
Combining these two bound completes the proof.

\section{Proof of \cref{thm:bride:of:frankinstein}}

This result is established from \cref{thm:ass:error} and some basic
properties $\QQ$ under the given assumptions. First, note that, cf.
\eqref{eq:ev:d:cond:sym:case}, (\ref{eq:EV:real:ass})
\begin{equation}\label{e:max_|lambda-j(Q)|=<|lambda-1(Q)|}
  |\lambda_d(\QQ)| \leq \max_{j =1, \ldots,d_{-}} |\lambda_{j}(\QQ)| \leq |\lambda_1(\QQ)|.
\end{equation}
Referring back to
\eqref{eq:frob:norm}, by \eqref{e:max_|lambda-j(Q)|=<|lambda-1(Q)|} we have the estimate
\begin{align}\label{eq:frob:sucks}
  \| \QQ e^{t\QQ} \|_F
  =&  \left( \sum_{j =1}^{d_{-}}
     | \lambda_j(\QQ) e^{ t  \lambda_j(\QQ)  } |^2 \right)^{1/2}
     \notag\\
  \leq& \sqrt{d_{-}} |\lambda_{1}(\QQ)| e^{-t |\lambda_{d_{-}}(\QQ)|}.
\end{align}
Similarly, from \eqref{eq:op:norm} and \eqref{e:max_|lambda-j(Q)|=<|lambda-1(Q)|}, we have that
\begin{align}\label{eq:op:better}
  \| \QQ e^{t\QQ} \|_{\op}
  \leq& \max_{j =1, \ldots,d_{-}}
        |\lambda_{j}(\QQ)| e^{-t |\lambda_{j}(\QQ)|}
        \notag\\
  \leq& |\lambda_{1}(\QQ)| e^{-t |\lambda_{d_{-}}(\QQ)|}.
\end{align}
Likewise, we observe that
\begin{align}\label{eq:thm1:extra:trms:1:F}
  \lVert \QQ\rVert^2_F = \sum_{k=1}^{d} \lambda_k(\QQ)^2 \, , \quad
  \lVert \QQ^+ \rVert^2_F = \sum_{k =1}^{d_-}\frac{1}{\lambda_k(\QQ)^2} \, ,
\end{align}
and that
\begin{align}\label{eq:thm1:extra:trms:1:F}
  \lVert \QQ\rVert^2_{\op} = \lambda_1(\QQ)^2 \, , \quad
  \lVert \QQ^+ \rVert^2_{\op} = \frac{1}{\lambda_{d_{-}}(\QQ)^2} \,.
\end{align}
Let $\UU$ be as in \eqref{eq:MP:inv:sym:case}. Noting, furthermore,
that $\II - \QQ \QQ^+ = \UU(\II - \II_{d_-})\UU^*$, where $\II_{d_-}$
is the matrix with $1$'s along the first $d_-$ diagonal elements and
zero otherwise, we have the bounds
\begin{align}\label{eq:proj:mat:bnds}
  \|\II - \QQ  \QQ^+\|_{F} = \sqrt{d - d_{-}},
  \quad
  \| \II - \QQ  \QQ^+\|_{\op} = 1.
\end{align}

Now, for any matrix norm $\|\cdot\|$, \eqref{eq:ass:error} implies
that
\begin{align}
    \| t(\II -\QQ^+ \QQ) &\JJ \QQ \QQ^{+}
                           +  \nabla_\JJ e^{t \QQ} - t  e^{t\QQ} \JJ\|
                           \notag\\
  &\leq C(1+ t)e^{-\kappa t}
    + 2\|   \QQ^+\| \|\JJ\| \|\II - \QQ \QQ^{+}\|
    \label{eq:ass:error_proof}
\end{align}    
for any $t \geq 0$, where $C > 0$ is given by
\eqref{eq:exp:error:ass:er}. Thus, for the choice of norm
$\|\cdot\| = \|\cdot\|_{F}$, by \eqref{eq:ass:error_proof},
\eqref{eq:frob:sucks}, \eqref{eq:thm1:extra:trms:1:F} and
\eqref{eq:proj:mat:bnds}, we obtain
\eqref{eq:bride:of:frankinstein:1}.  Similarly, for the choice of norm
$\|\cdot\| = \|\cdot\|_{\op}$, \eqref{eq:ass:error_proof}, in
conjunction with \eqref{eq:op:better}, \eqref{eq:thm1:extra:trms:1:F}
and \eqref{eq:proj:mat:bnds}, yields the bound
\eqref{eq:bride:of:frankinstein}.  Since
\eqref{eq:bride:of:frankinstein:1:rm:app} follows immediately from
\eqref{eq:bride:of:frankinstein:1}, and similarly between
\eqref{eq:bride:of:frankinstein:0:rm:app} and
\eqref{eq:bride:of:frankinstein}, the proof is now complete.

\section{Proof of \cref{t:lambda_n-1(Q)->-infty_non-symm_a.s.}}
\label{sec:Thm:rand:mat:main}

For use in this section, we introduce a probabilistic version of the
typical $O$ and $o$ asymptotic notations.  Given collections of random
variable $\{X_d\}_{d \in \bbN}$, $\{Y_d\}_{d \in \bbN}$, we write
\begin{equation}\label{e:X=O_a.s.(1)_Y=o_a.s.(1)}
X_d = O_{\as}(f(d)),\quad Y_d  = o_{\as}(f(d)), 
\end{equation}
for some $f: \bbN \to \RR^+$ to mean, respectively, that there exists
a random variable $C = C(\omega)$, not dependent on $d$, such
that $|X_d|/f(d) \leq C$ for all $d \in \bbN$ a.s., and also that
$\lim_{d \rightarrow \infty}Y_d/f(d) = 0$ a.s.

In everything that follows we will make use of the so-called
\textit{Weyl's inequalities} (see e.g.
\cite[Theorem 4.3.1, p.\ 239]{horn:johnson:2012}). Namely, let
${\mathbf A},{\mathbf B} \in {\mathcal S}(d,\bbC)$, each with its own
eigenvalues listed in increasing order as in our convention
\eqref{eq:EV:real:ass}. Then, for $i = 1,\hdots,d$,
\begin{equation}\label{e:Weyl_upper}
  \lambda_i({\mathbf A}+{\mathbf B}) \leq
  \lambda_{i+j}({\mathbf A}) + \lambda_{d-j}({\mathbf B}),
  \quad j = 0,1,\hdots,d-i,
\end{equation}
and
\begin{equation}\label{e:Weyl_lower}
  \lambda_{i-j+1}({\mathbf A}) + \lambda_j({\mathbf B})
  \leq \lambda_i({\mathbf A}+{\mathbf B}), \quad j = 1,\hdots,i.
\end{equation}
In what follows, we also make use of the fact that, for ${\mathbf A} \in {\mathcal S}(d,\bbC)$,
\begin{equation}\label{e:eigenvals_of_-A}
  \lambda_\ell(-{\mathbf A}) = - \lambda_{d-\ell+1}({\mathbf A}),
  \quad \ell = 1,\hdots,d.
\end{equation}

With these preliminaries now in hand, we turn to the proof of
\cref{t:lambda_n-1(Q)->-infty_non-symm_a.s.}.  This result is
established with the aide of four auxiliary results, \cref{l:||Q||/n=O(1/n^(1/2-eta)},
\cref{l:max_Yn/n^eps}, \cref{l:||Q_4||=o_P(1)} and \cref{l:lambda-2(Q3)>=1+1/n}, whose
precise statements and proof are provided immediately afterward.
\begin{proof}
  To demonstrate our desired result,
  \eqref{e:lambda_n-1(Q)->-infty_non-symm}, it suffices to establish
  the lower bound
  \begin{equation}\label{e:sigma^2_2(Q)>=_a.s._bound}
    \sigma^2_2(\QQ) \geq
    (d-1)^2 \Big\{ \Big(\frac{\mu \hspace{0.5mm} d}{d-1}\Big)^2
    + O_{\as}\Big(\sqrt{ \frac{\log d}{d}}\Big)\Big\}
  \end{equation}
  as well as the upper bound
  \begin{equation}\label{e:sigma^2_2(Q)=<_a.s._bound}
    \sigma^2_d(\QQ) \leq
    (d-1)^2 \Big\{ \Big(\frac{\mu \hspace{0.5mm} d}{d-1}\Big)^2
    + O_{\as}\Big(\sqrt{ \frac{\log d}{d}}\Big)\Big\}.
  \end{equation}
  With this in mind, we now decompose $\QQ$ as follows, starting with
  the case \eqref{e:RM_theorem_general_condition}.  Fix
  $d \in \bbN \backslash\{1\}$.  Taking $F_X$ to be the distribution
  defining the elements in \eqref{e:RM_theorem_general_condition}, draw
  \begin{align}\label{eq:diag:fill:trk}
    \{\widetilde{q}_{ii}\}_{i=1,\hdots,d}
    \stackrel{\textnormal{iid}}\sim F_X
    \text{ independently of } \{q_{ij}\}_{i,j=1,\hdots,d}.
   \end{align} 
   Here note carefully that $\widetilde{q}_{ii} \stackrel{d}= q_{ij}$
   and $\widetilde{q}_{ii} \stackrel{d}\neq q_{ii}$.  Now recast
   $\QQ = \{q_{ij}\}_{i,j}$ as 
\begin{align}
  & \begin{pmatrix}
0 & q_{12}-\mu & \hdots & q_{1d}-\mu \\
q_{21}-\mu & 0 & \hdots & q_{2d}-\mu \\
\vdots & \vdots & \ddots & \vdots \\
q_{d1}-\mu & q_{d2}-\mu & \hdots & 0 \\
\end{pmatrix}
+
\begin{pmatrix}
q_{11} & \mu & \hdots & \mu \\
\mu & q_{22} & \hdots & \mu \\
\vdots & \vdots & \ddots & \vdots \\
\mu & \mu & \hdots & q_{dd} \\
\end{pmatrix} \notag \\
  &=\left\{\!
    \begin{pmatrix}
\widetilde{q}_{11}-\mu & q_{12}-\mu & \hdots & q_{1d}-\mu \\
q_{21}-\mu & \widetilde{q}_{22}-\mu & \hdots & q_{2d}-\mu \\
\vdots & \vdots & \ddots & \vdots \\
q_{d1}-\mu & q_{d2}-\mu & \hdots & \widetilde{q}_{dd}-\mu \\
\end{pmatrix}\right.
  \label{eq:En:or:Wigner}\\
& \qquad \quad \quad +
\left.\begin{pmatrix}
-\widetilde{q}_{11}+\mu & 0 & \hdots & 0 \\
0 & -\widetilde{q}_{22}+\mu & \hdots & 0 \\
\vdots & \vdots & \ddots & \vdots \\
0 & 0 & \hdots & -\widetilde{q}_{dd}+\mu \\
\end{pmatrix}\!\right\}
  \notag \\
  &\quad\;+ (1-d)
    \left\{ \!
    \left(\! \!\!\!
      \begin{array}{*{24}{c@{\hspace{2pt}}}}
\frac{-q_{11}-(d-1)\mu}{d-1} & 0 & \hdots & 0 \\
0 & \frac{-q_{22}-(d-1)\mu}{d-1} & \hdots & 0 \\
\vdots & \vdots & \ddots & \vdots \\
0 & 0 & \hdots & \frac{-q_{dd}-(d-1)\mu}{d-1}\\
    \end{array}
  \! \!  \right)
  \right.
  \notag\\ 
  &\qquad \qquad \qquad \quad+  \left. 
        \begin{pmatrix}
\mu & \frac{-\mu}{d-1} & \hdots & \frac{-\mu}{d-1} \\
\frac{-\mu}{d-1} & \mu & \hdots & \frac{-\mu}{d-1} \\
\vdots & \vdots & \ddots & \vdots \\
\frac{-\mu}{d-1} & \frac{-\mu}{d-1} & \hdots & \mu \\
\end{pmatrix}
  \!\right\}
  \label{e:Q_basic_decomp}\\
  &=: \QQ_1 + \QQ_2 + (1-d) \big\{\QQ_3 +  \QQ_4\big\}
  =: \RM + (1-d) \QQ_4  \label{e:Q1,Q2,Q3}
\end{align}

Working from \eqref{e:Q1,Q2,Q3} we start by establishing the first
bound \eqref{e:sigma^2_2(Q)>=_a.s._bound} under
\eqref{e:RM_theorem_general_condition}. Noting from
\cref{l:lambda-2(Q3)>=1+1/n} that $\QQ_4$ is symmetric and positive,
making use of \eqref{e:Weyl_lower}, first with $i =2, j =1$, then with
$i=1, j =1$, invoking \eqref{e:eigenvals_of_-A} and finally using that
$\RM\RM^*$ is symmetric and non-negative, we find
\begin{align}
  &\sigma^2_2(\QQ) = \lambda_2(\QQ\QQ^*)
  \notag\\
 &\geq (d-1)^2 \Big\{ \lambda_{2}(\QQ^2_4) + \lambda_{1}\Big(\frac{\RM\RM^*}{(d-1)^2}
   - \frac{\QQ_4\RM^*+ \RM\QQ^*_4}{d-1} \Big) \Big\}
  \notag\\
 &\geq (d-1)^2 \Big\{ \lambda_{2}(\QQ^2_4) + \lambda_{1}\Big(\frac{\RM\RM^*}{(d-1)^2}\Big)
   +\lambda_{1} \Big(\frac{\QQ_4\RM^*+ \RM\QQ^*_4}{1-d} \Big) \Big\}
  \notag\\
 &\geq (d-1)^2 \Big\{ \lambda_{2}(\QQ^2_4) 
   -\lambda_{d} \Big(\frac{\QQ_4\RM^*+ \RM\QQ^*_4}{d-1} \Big) \Big\}.
   \label{eq:sig:2:lower:setup}
\end{align}
Now observe
\begin{align}
  |\lambda_j(\AAA)| \leq \sigma_d(\AAA) = \| \AAA\|_{op} \quad
  \text{ for
  any $j =1,\ldots, d$},
  \label{eq:dumb:op:norm:ob}
\end{align}
whenever $\AAA$ is symmetric.
Furthermore,
\begin{equation}\label{e:lambda_2(Q^2_3)}
\lambda_{2}(\QQ^2_4) = \lambda^2_{2}(\QQ_4) = \frac{\mu^2d^2}{(d-1)^2},
\end{equation}
where the equality follows from
\eqref{e:lambda-ell(Q3)=1+1/n,ell=2,...,n} in
\cref{l:lambda-2(Q3)>=1+1/n}.  Combining \eqref{eq:sig:2:lower:setup},
\eqref{eq:dumb:op:norm:ob} and \eqref{e:lambda_2(Q^2_3)}, we therefore
infer
\begin{equation}\label{e:sigma^2_2(Q)>=_bound}
  \sigma^2_2(\QQ) \geq (d-1)^2
  \Big\{ \Big( \frac{\mu \hspace{0.5mm} d}{d-1}\Big)^2 
- \frac{1}{d-1} \|\QQ_4\RM^*+ \RM\QQ^*_4\|_{\op}\Big\}.
\end{equation}

However, by  \eqref{e:||Q||/n=O(1/n^(1/2-eta)} in \cref{l:||Q||/n=O(1/n^(1/2-eta)},  \eqref{e:oas:Zd:con} in \cref{l:max_Yn/n^eps}
and \eqref{e:||Q_4||=o_P(1)} from \cref{l:||Q_4||=o_P(1)}
we have
\begin{align}
  \frac{d}{d-1}  \frac{\|\RM\|_{\op}}{d}
  \leq& \frac{d}{d-1} \left(\frac{\|\QQ_1\|_{\op}}{d} +
  \frac{\|\QQ_2\|_{\op}}{d}\right) + \|\QQ_3\|_{\op}
  \notag\\
       =& O_{\as}\Big(\sqrt{ \frac{\log d}{d}}\Big).
        \label{e:||Q-cal/n||=oP(1)}
\end{align}
Hence, by \eqref{e:||Q-cal/n||=oP(1)} and
\eqref{e:lambda-ell(Q3)=1+1/n,ell=2,...,n} in
\cref{l:lambda-2(Q3)>=1+1/n},
\begin{align}
  \frac{1}{d-1} &\|\QQ_4\RM^*+ \RM\QQ^*_4\|_{\op}
  \leq \frac{2d}{d-1} \hspace{0.5mm} \|\QQ_4 \|_{\op} \Big\|\frac{\RM}{d} \Big\|_{\op}
  \notag\\
 &=  O(1) \cdot O_{\as}\Big(\sqrt{ \frac{\log d}{d}}\Big)
  = O_{\as}\Big(\sqrt{ \frac{\log d}{d}}\Big).
  \label{e:(1/n)||Q3_Q*-cal+Q-cal_Q3||}
\end{align}
Thus, by \eqref{e:sigma^2_2(Q)>=_bound}, \eqref{e:||Q-cal/n||=oP(1)} and \eqref{e:(1/n)||Q3_Q*-cal+Q-cal_Q3||}, the lower bound \eqref{e:sigma^2_2(Q)>=_a.s._bound} holds.

We now turn to verify \eqref{e:sigma^2_2(Q)=<_a.s._bound}. Here, by
\eqref{e:Q1,Q2,Q3}, \eqref{e:Weyl_upper} with $j = 0, i = d$, 
\eqref{e:lambda_2(Q^2_3)} and finally \eqref{eq:dumb:op:norm:ob},
\begin{align*}
  &\sigma^2_d(\QQ)
  = \lambda_d(\QQ \QQ^*)
  \notag\\
   &\leq (d-1)^2 \Big\{ \lambda_{d}(\QQ^2_4)
     + \lambda_{d}\Big(\frac{{\mathbf R}{\mathbf R}^*}{(d-1)^2}\Big) +
     \lambda_d\Big(\frac{\QQ_4{\mathbf R}^*+ {\mathbf R}\QQ^*_4}{1-d} \Big) \Big\}
     \notag\\
  &\leq (d-1)^2 \Big\{ \Big(\frac{\mu \hspace{0.5mm} d}{d-1}\Big)^2
    + \Big\|  \frac{{\mathbf R}}{d-1}\Big\|^2_{\op}
    + \frac{1}{d-1} \|\QQ_4{\mathbf R}^*+ {\mathbf R}\QQ^*_4\|_{\op}\Big\}.
\end{align*}
The upper bound \eqref{e:sigma^2_2(Q)=<_a.s._bound} now follows from
\eqref{e:||Q-cal/n||=oP(1)} and \eqref{e:(1/n)||Q3_Q*-cal+Q-cal_Q3||}.
Thus, as a consequence of \eqref{e:sigma^2_2(Q)>=_a.s._bound} and
\eqref{e:sigma^2_2(Q)=<_a.s._bound},
\eqref{e:lambda_n-1(Q)->-infty_non-symm} holds under condition
\eqref{e:RM_theorem_general_condition}, as claimed.  We have
established the desired result in the first case \eqref{e:RM_theorem_general_condition}.

Now suppose that condition \eqref{e:RM_theorem_symmetry_condition}
holds. In order to establish \eqref{e:lambda_n-1(Q)->-infty_non-symm},
it suffices to replace $q_{ij}$, $i < j$, with $q_{ji}$ in expression
\eqref{eq:En:or:Wigner} and in the definition of $\QQ_3$.  We then
follow the rest of the argument for proving
\eqref{e:lambda_n-1(Q)->-infty_non-symm} under condition
\eqref{e:RM_theorem_general_condition} noting that
\cref{l:||Q||/n=O(1/n^(1/2-eta)}, \cref{l:||Q_4||=o_P(1)} also hold in
the symmetric case.  This concludes the proof.
\end{proof}


We turn now to the Lemmata \ref{l:||Q||/n=O(1/n^(1/2-eta)},
\ref{l:max_Yn/n^eps}, \ref{l:||Q_4||=o_P(1)} and
\ref{l:lambda-2(Q3)>=1+1/n}, which are used in the proof of
\cref{t:lambda_n-1(Q)->-infty_non-symm_a.s.}.  These results
correspond to the bounds we use for each of elements in the
decomposition \eqref{e:Q1,Q2,Q3}. We start off with
\cref{l:||Q||/n=O(1/n^(1/2-eta)} which essentially packages results
found in \cite{yin:bai:krishnaiah:1988}, \cite{bai:1999}.

\begin{lemma}\label{l:||Q||/n=O(1/n^(1/2-eta)}
  Let $\QQ_1 = \QQ_1(d)$, $d \in \bbN \setminus \{1\}$, be the
  sequence of random matrices defined as in \eqref{eq:En:or:Wigner}.
  Here we suppose that the off diagonal elements $q_{ij}$ defining
  $\QQ_1$ are either as in \eqref{e:RM_theorem_general_condition} or
  as in \eqref{e:RM_theorem_symmetry_condition} and that the diagonal
  elements $\widetilde{q}_{ii}$ are drawn as in
  \eqref{eq:diag:fill:trk}.  Then, in either case, we have
  \begin{equation}\label{e:||Q||/n=O(1/n^(1/2-eta)}
    \frac{\|\QQ_1\|_{\op}}{d} = O_{\as}\Big( \frac{1}{\sqrt{d}}\Big).
  \end{equation}
\end{lemma}
\begin{proof}
  To establish \eqref{e:||Q||/n=O(1/n^(1/2-eta)} in the first case,
  \eqref{e:RM_theorem_general_condition}, we note that
  $d^{-1}\QQ_1 \QQ^*_1$ forms a sample covariance matrix, where each
  entry of $\QQ_1$ is centered and has infinitely many moments; cf.
  \cref{rmk:exp:mom}. Thus, by \cite[Theorem 3.1, p.\
  517]{yin:bai:krishnaiah:1988},
\begin{equation}\label{e:|Q-tilde_1|=sqrt(n)OP(1)}
  \|\QQ_1\|_{\op} = \sqrt{d}\cdot \sqrt{\lambda_{d}\Big(\frac{\QQ_1 \QQ^*_1}{d}\Big)}
  = \sqrt{d}\cdot \sqrt{4 \sigma^2 + o_{\as}(1)},
\end{equation}
which immediately yields \eqref{e:||Q||/n=O(1/n^(1/2-eta)}.

In the second case, where all the off-diagonal elements are determined
starting from \eqref{e:RM_theorem_symmetry_condition}, $\QQ_1$ is a
centered (mean-zero) Wigner matrix in the typical nomenclature
followed by \cite{bai:1999}.  Thus, noting \eqref{eq:op:norm} and
that, as in the previous case, $\QQ_1$ has infinitely many moments
\eqref{e:||Q||/n=O(1/n^(1/2-eta)} is thus a direct consequence of
\cite[Theorem 2.12, p.\ 630]{bai:1999}.
\end{proof}

We next state and establish \cref{l:max_Yn/n^eps}. Note that the
assumption of independence among the random variables is not needed
for this result.  Here and below in what follows we adopt the useful
notational convention that $o_{\textnormal{a.s.}}(1)$ denotes a random
variable that vanishes almost surely, as $d\to \infty$ (cf.\
\eqref{e:X=O_a.s.(1)_Y=o_a.s.(1)}).
\begin{lemma}\label{l:max_Yn/n^eps}
  Let $\{Y_{d}\}_{d \in \bbN}$ be an identically distributed sequence
  of sub-exponential variables with mean $\mu$;
  cf. \cref{def:sub:exp:RV}. Let
  \begin{align}
    Z_{d} = \max_{k=1,\hdots,d} |Y_k - \mu|, d \in \bbN.
    \label{eq:max:sub:exp}
  \end{align} 
  Then, for any fixed $\varepsilon_0 > 0$, 
  \begin{align}\label{e:oas:Zd:con}
   \frac{Z_d}{d^{\varepsilon_0}} = o_{\as}(1).
  \end{align}
   In particular, for the sequence of diagonal matrices
   $\QQ_2 \equiv \QQ_2(d)$ given as in \eqref{e:Q1,Q2,Q3}, and any
   $\varepsilon_0 > 0$,
\begin{align}
  \frac{\|\QQ_2\|_{op}}{d^{\varepsilon_0}} = o_{\as}(1).
  \label{e:QQ2=oP(1)}
\end{align}
\end{lemma}
\begin{proof}
For any $M > 0$,
\begin{align*}
  \bbP\Big(\frac{Z_d}{d^{\varepsilon_0}} \geq \frac{M}{d^{\varepsilon_0/2}}\Big)
  &= \bbP\Big(\bigcup^{d}_{k=1}\Big\{\frac{|Y_k - \mu|}{d^{\varepsilon_0}}
  \geq \frac{M}{d^{\varepsilon_0/2}}\Big\}\Big) \\
  &\leq d \cdot\bbP\big(|Y_1 - \mu| \geq d^{\varepsilon_0/2}M\big) 
 \leq  2d\cdot e^{- \frac{d^{\varepsilon_0/2}M}{K}}.
\end{align*}
In the first and second inequalities, respectively, we use the
assumption that the random variables $\{Y_{d}\}_{d \in \bbN}$ are
identically distributed as well as relation \eqref{e:def_subexp_RV}
for some suitable value $K > 0$. Therefore,
$\sum^{\infty}_{d=1}\bbP(Z_d/d^{\varepsilon_0/2} \geq
M/d^{\varepsilon_0/2})< \infty.$ Since
$\lim_{d \rightarrow \infty}M/d^{\varepsilon_0/2} = 0$,
\eqref{e:oas:Zd:con} now follows as a consequence of the
Borel-Cantelli lemma.

Regarding the second claim observe that, cf. \eqref{eq:diag:fill:trk},
$\QQ_2$ is a diagonal matrix containing identically distributed
sub-exponential random variables.  This mean that, taking
$Y_k :=\tilde{q}_{kk}$, $\|\QQ_2\|$ is of the form
\eqref{eq:max:sub:exp}. As such the second claim \eqref{e:QQ2=oP(1)}
now follows from the first \eqref{e:oas:Zd:con}, completing the proof.
\end{proof}

Turning to our bounds on $\|\QQ_3\|_{op}$, we have the following lemma.
\begin{lemma}\label{l:||Q_4||=o_P(1)} Let $\QQ_3 \equiv \QQ_3(d)$,
  $d \in \bbN$, be a sequence of random matrices as in
  \eqref{e:Q1,Q2,Q3}. Here we assume that the $q_{ii}$ are determined
  either according to \eqref{e:RM_theorem_general_condition} or
  according to \eqref{e:RM_theorem_symmetry_condition}. Then, in
  either of these cases,
    \begin{equation}\label{e:||Q_4||=o_P(1)}
      \|\QQ_3\|_{\op} = O_{\as}\Big(\sqrt{ \frac{\log d}{d}}\Big).
    \end{equation}
\end{lemma}
\begin{proof}
  Under either \eqref{e:RM_theorem_general_condition} or
  \eqref{e:RM_theorem_symmetry_condition}, each of the entries along
  the main diagonal is a renormalized sums of iid sub-exponential
  random variables i.e.
  \begin{align}
    \frac{-q_{ii} - \mu(d-1)}{d-1}
    = \frac{1}{d-1} \sum_{j \not =i}^d (q_{ij} - \mu)
  \end{align}
  While these diagonal elements are not independent under
  \eqref{e:RM_theorem_symmetry_condition} note carefully
  that we do not use the independence of the rows of $\QQ_3$
  in the arguments that follow.

  Now, by Bernstein's inequality,
  \cite[p.\ 29, Theorem 2.8.1]{vershynin:2018}, there exists a
  constant $C > 0$ such that, for any $\varepsilon > 0$,
  \GD{\begin{align}
        \bbP&\big(\big|-q_{ii}-\mu(d-1)\big| > \varepsilon\big)\notag\\
            &\leq 2 \cdot \exp\Big\{ - C \min\Big\{
              \frac{\varepsilon^2}{(d-1)\|X_\mu\|^{2}_{\psi_1}},
              \frac{\varepsilon}{\|X_\mu\|_{\psi_1}} \Big\}\Big\},
    \label{e:P(max|-qii-n|/n>eps)}
\end{align}
for $i = 1, \ldots, d$ where $X_{\mu}:= X - \mu$ for $X \sim F_X$ and
the right-hand side of \eqref{e:P(max|-qii-n|/n>eps)} involves the
sub-exponential norm \eqref{e:sub-exp_norm}. Fix $\delta > 0$, and let
$\eta:= (2 + \delta)/C$ where $C$ is the constant in the upper
bound in \eqref{e:P(max|-qii-n|/n>eps)}. Hence, for all
$d \in \bbN \backslash \{1\}$,
\begin{align}\label{e:P(|Q3|>eps)}
  \bbP&\Big( \sqrt{\frac{d-1}{\eta \log d}}
        \hspace{0.5mm}\|\QQ_3(d)\|_{\op}
        > \|X_\mu\|_{\psi_1}\Big)\notag\\
      &= \bbP\Big(
        \max_{i=1,\hdots,d}\Big|\frac{-q_{ii}-\mu(d-1)}{d-1}\Big|
        > \|X_\mu\|_{\psi_1}\sqrt{\frac{\eta \log d}{d-1}}\Big) \notag\\
      &= \bbP\Big(\bigcup^{d}_{i=1}\big\{\big|-q_{ii}-\mu(d-1)\big|
        > \|X_\mu\|_{\psi_1}\sqrt{\eta (d-1)\log d}\big\}\Big) \notag \\ 
      &\leq \sum^{d}_{i=1}\bbP\Big(\big|-q_{ii}-\mu(d-1)\big|
        > \|X_\mu\|_{\psi_1}\sqrt{\eta (d-1)\log d}\Big) \notag \\
      & \leq  2d \cdot \exp\big\{ -  C \min\big\{ \eta  \log d ,
        \sqrt{ \eta (d-1) \log d } \big\}\big\}.
\end{align}}
In the second inequality in \eqref{e:P(|Q3|>eps)},
we use \eqref{e:P(max|-qii-n|/n>eps)} with
$\varepsilon := \|X_\mu\|_{\psi_1} \sqrt{\eta (d-1)\log d}$
and the fact that $\{q_{ii}\}_{i=1,\hdots,d}$ are identically
(but not necessarily independently) distributed. Therefore,
for every $d$ sufficiently large,
\begin{align*}
\bbP\Big( \sqrt{\frac{d-1}{\eta \log d}}\|\QQ_3(d)\|_{\op}
> \|X_\mu\|_{\psi_1}\Big)\leq  \frac{2}{d^{1 + \delta}},
\end{align*}
and as a consequence,
\begin{align*}
\sum^{\infty}_{d=1}\bbP\big(\|\QQ_3(d)\|_{\op} >
  \|X_\mu\|_{\psi_1}\sqrt{\frac{\eta \log d}{d-1}}\big)< \infty.
\end{align*}
Hence, by the Borel-Cantelli lemma,
\begin{align*}
  \bbP\Big(\|\QQ_3(d)\|_{\op} >
  \|X_\mu\|_{\psi_1}\sqrt{\frac{\eta \log d}{d-1}}
  \textnormal{ i.o.}\Big)  = 0,
\end{align*}
so that the relation \eqref{e:||Q_4||=o_P(1)} holds. The
proof is complete.
\end{proof}

Finally we conclude with \cref{l:lambda-2(Q3)>=1+1/n} as follows.
\begin{lemma}\label{l:lambda-2(Q3)>=1+1/n}
  Let $\QQ_4 \equiv \QQ_4(d)$, $d \in \bbN$, be the sequence of
  symmetric matrices defined in \eqref{e:Q1,Q2,Q3}. Then, for any
  $d \in \bbN \backslash\{1\}$,
\begin{equation}\label{e:lambda-1(Q3)>=1/n}
  \lambda_1(\QQ_4) = 0,
\end{equation}
and
\begin{equation}\label{e:lambda-ell(Q3)=1+1/n,ell=2,...,n}
  \lambda_\ell(\QQ_4) = \GD{\frac{\mu \hspace{0.5mm} d}{d-1} }, \quad
  \text{ for } \ell = 2,\hdots,d.
\end{equation}
\end{lemma}
\begin{proof}
  The statement \eqref{e:lambda-1(Q3)>=1/n} is a consequence of the
  fact that
  $\QQ_4 \hspace{0.5mm}(1,\hdots,1)^\top = {\mathbf 0} \in \bbR^n$. To
  establish \eqref{e:lambda-ell(Q3)=1+1/n,ell=2,...,n}, it suffices to
  prove that
\begin{equation}\label{e:lambda-2(Q3)>=1+1/n}
  \frac{d}{d-1} \leq \GD{\frac{\lambda_2(\QQ_4)}{\mu}
    \leq \frac{\lambda_d(\QQ_4)}{\mu}} \leq \frac{d}{d-1},
  \quad d \in \bbN \backslash\{1\}.
\end{equation}
Let ${\mathbf 1} \in {\mathcal S}(d,\bbR)$ be a matrix of
ones and recast
\begin{equation}\label{e:Q3/mu}
\frac{1}{\mu}\QQ_4 = \II + \frac{1}{d-1}(\II-{\mathbf 1}). 
\end{equation}
By \eqref{e:Q3/mu}, \eqref{e:Weyl_lower} with $i = 2, j=1$
and \eqref{e:eigenvals_of_-A},
\begin{align}
  \frac{\lambda_2(\QQ_4)}{\mu}
  &\geq \lambda_{1}(\II) + \lambda_{2}\Big(\frac{1}{d-1}(\II-{\mathbf 1})\Big)
    \notag\\
  &= \lambda_{1}(\II) - \frac{1}{d-1}\lambda_{d-1}({\mathbf 1}-\II).
    \label{e:lambda-2(Q3)>=_Weyl_bound}
\end{align}
However, by \eqref{e:Weyl_upper}, in the case $i = d-1, j = 1$,
\begin{equation}\label{e:lambda_n-1(1-I)=-1}
  \lambda_{d-1}({\mathbf 1}-\II)
  \leq \lambda_{d-1}({\mathbf 1}) + \lambda_d(-\II) = -1,
\end{equation}
where we used that
\begin{equation}\label{e:rank(1)=1}
  \lambda_{1}(\mathbf{1}) = \cdots = \lambda_{d-1}({\mathbf 1}) = 0.
\end{equation}
This latter claim follows immediately from the fact that
$\textnormal{rank}({\mathbf 1}) = 1$ and
$\mathbf{1} \hspace{0.5mm}(1,\hdots,1)^\top= d(1,\hdots,1)^\top$.  The
first inequality in \eqref{e:lambda-2(Q3)>=1+1/n} is now a consequence
of \eqref{e:lambda-2(Q3)>=_Weyl_bound} and
\eqref{e:lambda_n-1(1-I)=-1}.

On the other hand, again by \eqref{e:Q3/mu}, \eqref{e:Weyl_upper}, this
time with $i = d, j = 0$, and \eqref{e:eigenvals_of_-A},
\begin{align}
  \GD{\frac{\lambda_d(\QQ_4)}{\mu}}
  &\leq \lambda_d(\II)
         + \lambda_d\Big(\frac{1}{d-1}(\II - {\mathbf 1})\Big)
    \notag\\
  &= \lambda_d(\II) - \lambda_1\Big(\frac{1}{d-1}({\mathbf 1}-\II)\Big).
        \label{e:lambda_n(Q3)=<_upper}
\end{align}
However, by \eqref{e:Weyl_lower} and \eqref{e:rank(1)=1},
\begin{equation}\label{e:lambda_1(1-I)>=_lower}
\lambda_1({\mathbf 1}-\II) \geq \lambda_1({\mathbf 1}) + \lambda_1(-\II)= -1.
\end{equation}
The third inequality in \eqref{e:lambda-2(Q3)>=1+1/n} is now a
consequence of \eqref{e:lambda_n(Q3)=<_upper} and
\eqref{e:lambda_1(1-I)>=_lower}. This establishes
\eqref{e:lambda-2(Q3)>=1+1/n} and, hence,
\eqref{e:lambda-ell(Q3)=1+1/n,ell=2,...,n}, completing the proof.
\end{proof}

\section{Surrogate-trajectory Hamiltonian Monte Carlo}
\label{sec:Sur:trad:HMC}

\newcommand{\G}{\mathbf{G}}

Hamiltonian Monte Carlo (HMC) \cite{duane1987hybrid,neal2011mcmc} is
an advanced MCMC procedure that uses numerical approximations to
Hamiltonian trajectories to generate Metropolis-type
\cite{metropolis1953equation} proposals far away from the current
Markov chain state.  On the one hand, this approach to proposal
generation helps reduce autocorrelation between chain states and is
particularly helpful within higher-dimensional state spaces.  On the
other hand, numerical integration of Hamilton's equations requires
repeated evaluation of the Hamiltonian potential energy's gradient,
and these repeated floating-point operations may become
computationally burdensome.

Lets briefly recall this HMC approach to resolve a given density
$\pi(\cdot)$ of a `target' probability measure of
interest.\footnote{In the Bayesian inference context, we are typically
  considering the posterior density function
  $\pi(\ttheta):=p(\ttheta|\Y)$ for the parameter
  $\ttheta\in\mathbb{R}^K$ given observed data $\Y$.}  We proceed by
considering a potential energy of the form
$U(\ttheta):=-\log \pi(\ttheta)$.  We then select an associated
kinetic energy $V(\mom) := | \G^{-1/2} \mom|^2$ for an appropriately
chosen symmetric-positive-definite mass matrix $\G$ (which is often
taken as the identity for simplicity).  One then observes that the
Gibbs measure, proportional to $e^{-H(\ttheta, \mom)}$ where
$H = U +V$, is invariant under the associated Hamiltonian dynamic.
Here note that the $\ttheta$ marginal of $e^{-H(\ttheta, \mom)}$
coincides with $\pi$ while the $\mom$ marginal is normally distributed
as $\mathcal{N}(\Zero,\G)$.

One operationalizes these observations as an algorithmic sampling
procedure as follows. At each step, given a current sample
$\ttheta^{(n)}$, one draws $\mom^{(n)} \sim \mathcal{N}(\Zero,\G)$.
From this initial state
$(\ttheta(0), \mom(0)) := (\ttheta^{(n)},\mom^{(n)})$ one then
numerically approximates the associated Hamiltonian dynamics using a
St\"{o}rmer-Verlet (velocity Verlet) or leapfrog integrator up to a
total integration time $\tau >0$ and using integration step size
$\epsilon > 0$.  In this context, note that a single iteration of this
integrator takes the form, \cite{leimkuhler2004simulating},
\begin{align}
  \mom\left(s + \frac{\epsilon}{2}\right)
  &:= \mom(s) + \frac{\epsilon}{2} \nabla \log \pi(\ttheta(s)),  
	 \nonumber \\
  \ttheta(s+\epsilon) &:= \ttheta(s) + \epsilon \, \G^{-1}\mom(s+\frac{\epsilon}{2}), 
		 \label{eq:leap}\\ 
  \mom(s + \epsilon)
  &:= \mom\left(s+\frac{\epsilon}{2}\right)
    + \frac{\epsilon}{2} \nabla \log \pi(\ttheta(s+\epsilon))  \, .
	\nonumber
\end{align}
In this fashion the proposed new state is given through
\eqref{eq:leap} by $\ttheta(\tau)$.  To remove bias this procedure can
be augmented with an acceptance probability of the form
$\alpha^{(n)} := \exp( H( \ttheta(0), \mom(0)) - H(\ttheta(\tau),
\mom(\tau)) \wedge 1$.

Different strategies aim to speed up the leapfrog integrator's many
log-posterior gradient evaluations $\nabla \log \pi (\ttheta)$, as
these numerical routines often represent the algorithm's computational
bottleneck.  In model-specific contexts,
\cite{holbrook2021massive,holbrook2022bayesian,holbrook2022viral}
yield parallelization strategies, and \cite{ji2020gradients} develops
dynamic programming techniques, to accelerate the evaluation of
$\nabla \log \pi (\ttheta)$.

A small body of work considers another approach by replacing
$\nabla \log \pi (\ttheta)$ with a suitable approximation
$\widetilde{\nabla} \log \pi (\ttheta)$ and recognizing that the
modified \eqref{eq:leap} continues to satisfy path reversibility and
volume preservation, two essential ingredients for well-specified
HMC. It follows immediately that the resulting `surrogate trajectory
HMC' continues to sample the correct target distribution $\pi(\cdot)$;
see further details \cite{neal2011mcmc, glatt2020accept}.
Nonetheless, the acceptance rates and overall efficiency of such
samplers may suffer when approximations are poor.

The majority of surrogate HMC methods first obtain a small sample of
exact gradient evaluations and then use some model to interpolate:
\cite{zhang2017precomputing} assumes the approximate gradient follows
a piecewise constant form across a grid;
\cite{rasmussen2003gaussian,lan2016emulation} construct approximations
using Gaussian processes; and \cite{zhang2017hamiltonian,li2019neural}
do the same using neural networks.  But an even simpler approach to
surrogate HMC may be appropriate when gradients have series
representations as we can leverage here in \eqref{eq:grad:series}.

\section{Visualizing the posterior mean random effects}\label{sec:randomEffs}

Figure \ref{fig:randEffs} displays posterior means of the exponentiated random effects, which one may interpret as multiplicative deviations from the fixed effects' contributions to the generator matrix.  Whereas the vast majority of exponentiated random effects have posterior means close to 1, indicating no deviation from the fixed-effect model, a few exhibit posterior means that are significantly greater than 1.  In particular, the rate element corresponding to transfer from the US to Hubei, CN, exhibits a 1.81-fold random-effect derived increase to the fixed-effect component.  This large multiplicative increase agrees with, but goes beyond, the influence that the Hubei asymmetry holds for the entire generator matrix model. 

\setcounter{figure}{0}
\renewcommand{\thefigure}{S\arabic{figure}}
\begin{figure}[t]
	\centering
	\includegraphics[width=0.8\linewidth]{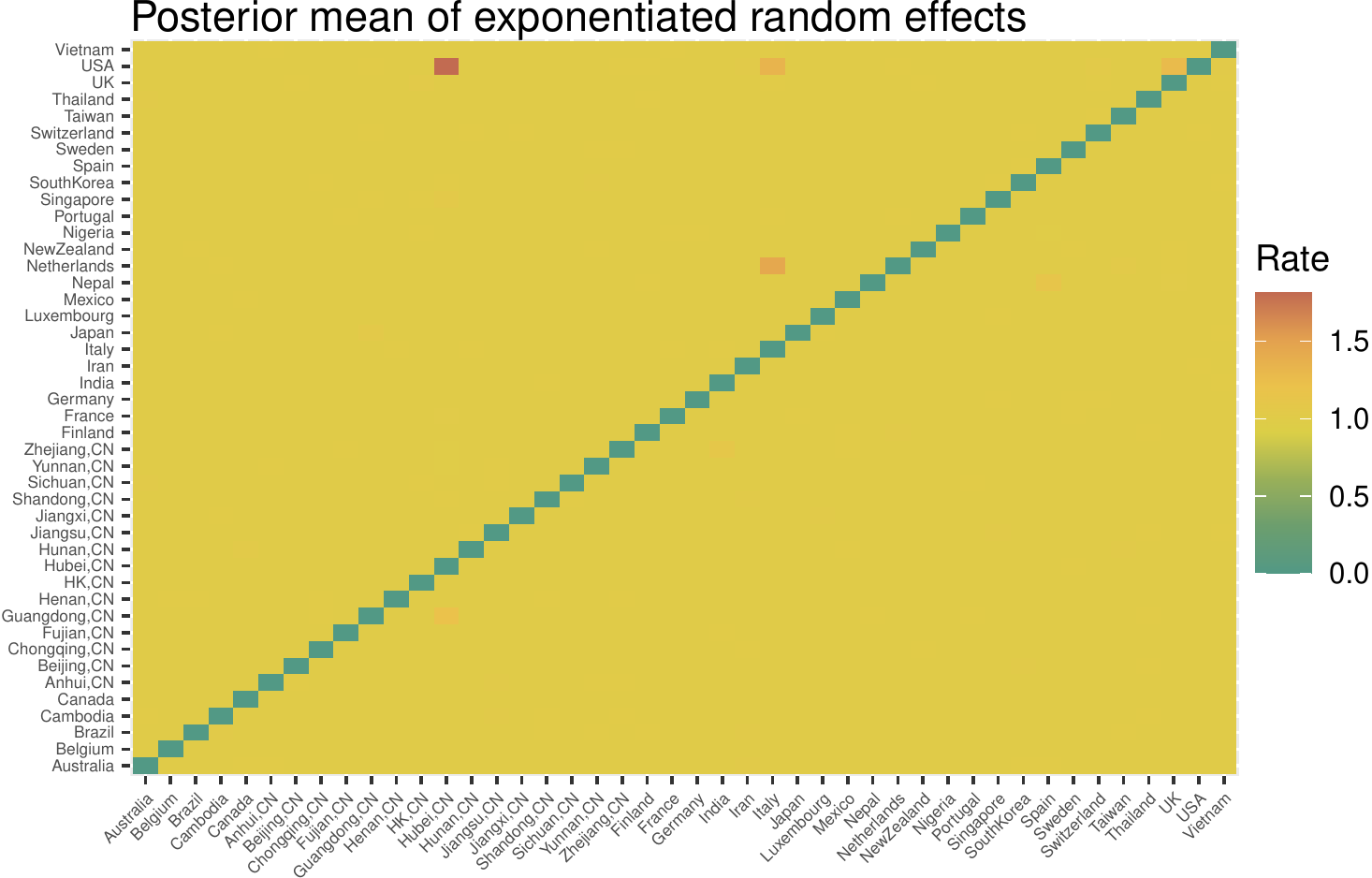}
	\caption{Posterior means for exponentiated random effects convey expected multiplicative deviations from the portion of the rate attributable to fixed effects for each corresponding element of the generator matrix.  Notably, we infer a roughly 1.81-fold posterior mean increase in the rate of transitions from the US to Hubei, CN, beyond that portion of the rate which may be explained by fixed effects.  Less pronounced are posterior mean multiplicative increases of 1.34 (US to Italy), 1.27 (US to UK), 1.44 (Netherlands to Italy) and 1.21 (Guangdong, CN, to Hubei, CN).}\label{fig:randEffs}
\end{figure}


\end{document}